\documentclass[11pt]{article}
\usepackage{cite}
\usepackage{amsmath,amsfonts,amssymb}
\pdfoutput=1
\usepackage[small,bf,hang]{caption}
\usepackage{slashed}
\input epsf.sty
\usepackage{epsfig}
\usepackage[titletoc,toc]{appendix}
\usepackage{xcolor}

\graphicspath{ {./images/} }

\def\hybrid{
        \topmargin -20pt
        \oddsidemargin 0pt
        \headheight 0pt \headsep 0pt
        \textwidth 6.55in 
        \textheight 9.5in 
        \marginparwidth .875in
        \parskip 5pt plus 1pt \jot = 1.5ex}

\hybrid

 \csname
@addtoreset\endcsname{equation}{section}

\newcommand{\p}{\partial}

\def\moth{\mathsurround=0pt}
\newdimen\zo \zo=0pt

\def\tick{\leaders\hrule height 0.5ex depth 0pt \hskip 0.5pt}
\def\upboxfill{$\moth \setbox\zo\hbox{\tick}%
  \hskip 3pt\hbox to 0pt{$\tick$\hss}\hrulefill \hbox to 7.5pt{$\tick$\hss}$}

\def\dtick{\leaders\hrule height .34pt depth 0.5ex \hskip 0.5pt}
\def\downboxfill{$\moth \setbox\zo\hbox{\dtick}%
  \hskip 2pt\hbox to 0pt{$\dtick$\hss}\hrulefill \hbox to 2pt{$\dtick$\hss}$}

\def\bec{\begin{center}}
\def\ec{\end{center}}

\def\tr{{\rm tr}}

\def\be{\begin{equation}}
\def\ee{\end{equation}}
\def\bea{\begin{eqnarray}}
\def\eea{\end{eqnarray}}
\def\ba{\begin{array}}
\def\ea{\end{array}}

\usepackage{hyperref}

\usepackage{amsthm}

\theoremstyle{definition}

\theoremstyle{theorem}
\newtheorem{claim}{Claim}

\begin{document}

\begin{titlepage}
	
	\rightline{\tt MIT-CTP-5640}
	\hfill \today
	\begin{center}
		\vskip 0.5cm
		
		{\Large \bf {String vertices for the large $N$ limit}
		}
		
		\vskip 0.5cm
		
		\vskip 1.0cm
		{\large {Atakan Hilmi F{\i}rat}}
		
		{\em  \hskip -.1truecm
			Center for Theoretical Physics \\
			Massachusetts Institute of Technology\\
			Cambridge, MA 02139, USA\\
			\tt \href{firat@mit.edu}{firat@mit.edu} \vskip 5pt }

		\vskip 1.5cm
		{\bf Abstract}
		
	\end{center}
	
	\vskip 0.5cm
	
	\noindent
	\begin{narrower}
		\noindent
		String vertices of open-closed string field theory on an arbitrary closed string background with $N$ identical D-branes are investigated when $N$ is large. We identify the relevant geometric master equation and solve it using open-closed hyperbolic genus zero string vertices with a milder systolic constraint. The limits corresponding to integrating out open or closed strings are investigated. We highlight the possible implications of our construction to AdS/CFT correspondence.
	\end{narrower}
	
\end{titlepage}

\tableofcontents

\section{Introduction}

Richness of defining covariant string field theory (SFT) using different field parametrizations has been a blessing and a curse at the same time.\footnote{For reviews of different aspects of SFT, see~\cite{Zwiebach:1992ie,deLacroix:2017lif,Erler:2019loq,Erbin:2021smf,Okawa:2012ica,Erler:2019vhl,Maccaferri:2023vns}. We also highlight the recent developments~\cite{Sen:2019qqg,Sen:2020cef,Sen:2020oqr,Sen:2020ruy,Sen:2020eck,Sen:2021qdk,Sen:2021tpp,Sen:2021jbr,Alexandrov:2021shf,Alexandrov:2021dyl,Agmon:2022vdj,Eniceicu:2022nay,Alexandrov:2022mmy,Eniceicu:2022dru,Chakravarty:2022cgj,Eniceicu:2022xvk,Erler:2020beb,Erbin:2021hkf,Scheinpflug:2023osi,Scheinpflug:2023lfn,Maccaferri:2022yzy,Maccaferri:2023gcg,Okawa:2022sjf,Konosu:2023pal,Konosu:2023rkm,Grigoriev:2021bes,Doubek:2020rbg,Berkovits:2021eny,Cho:2023khj,Manki}.} This choice is intimately connected to various mathematical structures on Riemann surfaces, such as quadratic differentials~\cite{saadi1989closed, Erbin:2022rgx}, minimal-area metrics~\cite{Zwiebach:1990nh,Headrick:2018dlw,Headrick:2018ncs,Naseer:2019zau}, hyperbolic metrics~\cite{Moosavian:2017qsp, Moosavian:2017sev, Costello:2019fuh, Cho:2019anu, Firat:2021ukc,Wang:2021aog, Ishibashi:2022qcz}, Liouville theory~\cite{Firat:2023glo, Firat:2023suh}, and complex projective structures~\cite{Costello:2019fuh, dumas2009complex}. These distinct ways of approaching the same problem have illuminated many corners of SFT and the theory of Riemann surfaces in the past.

However, among all this plethora of choices, the most natural parametrization for string theory in its full generality is still not clear. Nevertheless, the answer is obvious in certain regimes: it is advantageous to consider the cubic theory constructed using Witten's vertex as far as classical open strings are concerned~\cite{Witten:1985cc}. Having a cubic theory is the primary reason why the open SFT is solvable. Maybe not so obviously, polyhedral vertices defined by Strebel differentials appear to be the most natural choice for the classical closed SFT, evidenced by having the best possible convergence behavior for the effective tachyon potential in bosonic closed SFT~\cite{Belopolsky:1994sk,Belopolsky:1994bj}. It is widely held belief among practitioners that this parametrization has the best chance of providing solutions to closed SFT in future. In fact, there is a Strebel differential underlying Witten's vertex, so polyhedral vertices are the natural generalization of Witten's construction for open strings to closed strings. 

Unfortunately the situation becomes murky beyond classical theories and/or upon coupling closed and open strings. It is an important problem to find the best possible way to parametrize interactions of SFT in different regimes to make off-shell computations as efficient as possible---even though the theory would still be non-polynomial in general~\cite{Sonoda:1989sj}. In this note, we propose a set of string vertices for the (oriented, bosonic) open-closed SFT in~\textit{the large $N$} limit based on open-closed hyperbolic string vertices of Cho~\cite{Cho:2019anu}, which appears to be a quite natural choice for the SFT in this regime. The large $N$ limit, or~\textit{the planar limit}, here refers to the open-closed SFT in the presence of large number of identical D-branes. The SFT in this limit has recently been investigated by Maccaferri, Ruffino, and Vo\v{s}mera from the algebraic point of view~\cite{Maccaferri:2023gof} and they argued that it is sufficient to consider genus zero surfaces with arbitrary number of open string boundaries and bulk/boundary punctures. Here we elucidate the geometric counterparts of their results, which implicates the relevant master equation is
\begin{align}
	\p \mathcal{V}_{pl} + {1 \over 2} \{\mathcal{V}_{pl}, \mathcal{V}_{pl} \} + \Delta_o^{(1)}  \mathcal{V}_{pl}  = 0 \, ,
\end{align}
for string vertices in the planar limit $\mathcal{V}_{pl}$, refer to~\eqref{eq:PlanarBV}. This geometric recursion relation was originally noted in~\cite{Zwiebach:1997fe} without elaboration nor acknowledging its significance to the large $N$ limit.

\begin{figure}[t]
	\centering
	\includegraphics[width=0.56\textwidth,height=0.53\textwidth]{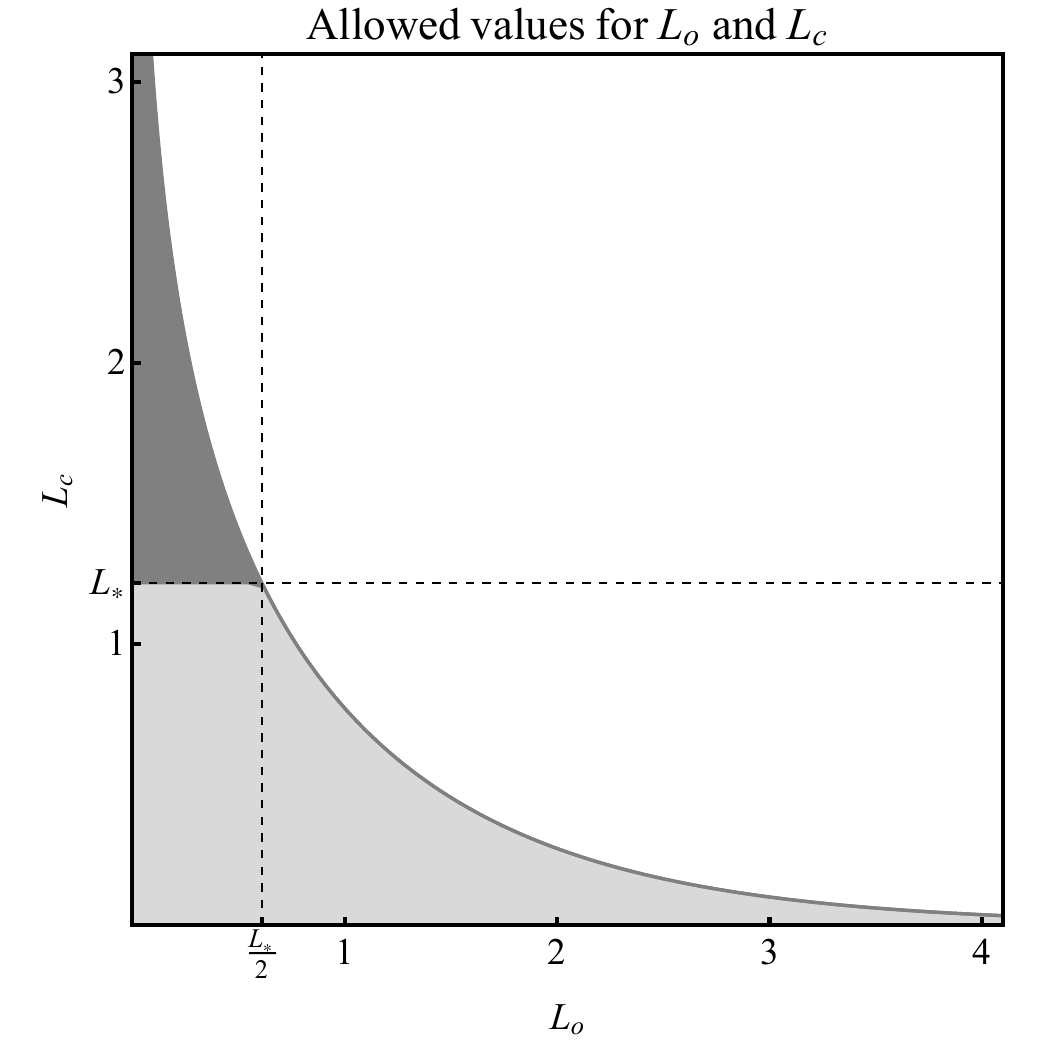}
	\caption{\label{fig:Constraint}Allowed values for the border lengths $L_c, L_o$ for open-closed hyperbolic vertices (light gray region)~\cite{Cho:2019anu}. Here $L_\ast \approx 1.22$, see~\eqref{eq:Condition}. In the large $N$ limit we are allowed to use $L_c > L_\ast$ as well (dark gray region).}
\end{figure}
To summarize our results, we construct a family of open-closed hyperbolic string vertices (with genus zero) parameterized by two positive numbers $L_c, L_o$ obeying the constraint
\begin{align}
	\sinh {L_c} \, \sinh {L_o} \leq 1 \, ,
\end{align}
in the large $N$ limit, see figure~\ref{fig:Constraint}. Here $L_c$ and $L_o$ stand for the lengths of the borders associated with closed and open strings respectively, see section~\ref{sec:first}. It is worth mentioning that there is no constraint on $L_c$ unlike the general open-closed hyperbolic SFT~\cite{Cho:2019anu}. So we investigate the limit $L_c \to 0$, $L_o \to \infty$ ($L_c \to \infty$, $L_o \to 0$, respectively), which corresponds to the situation for which open (closed) SFT deformed by couplings with on-shell closed (open) strings. We argue these string vertices are polyhedral in either case---i.e. they are characterized by Strebel differentials.\footnote{An introduction to Strebel differentials can be found in~\cite{strebel1984quadratic}. We employ terminology used in~\cite{Erbin:2022rgx}.}
We observe the former limit leads to Witten's open SFT deformed by the Ellwood invariant~\cite{zwiebach1992interpolating,Ellwood:2008jh}, while the latter limit is novel and leads to an explicit construction of closed SFT deformed by the tadpole considered in~\cite{Maccaferri:2023gof}.

This note consists of two sections, conclusion, and two appendices. In section~\ref{sec:first}, we summarize the basics of open-closed string vertices, including the geometric master equation and hyperbolic construction by Cho~\cite{Cho:2019anu}. Here we also briefly comment on the derivation of local coordinates for these vertices by generalizing the formalism of~\cite{Firat:2023glo} to open Riemann surfaces with punctures. In section~\ref{sec:N} we take the large $N$ limit for string vertices and introduce~\textit{the geometric planar master equation}. We subsequently solve it using hyperbolic geometry and explore the limits of vertices corresponding to parametrizations for which closed or open strings are taken on-shell. We discuss future prospects of our work in conclusion~\ref{sec:conclusion}. In appendix~\ref{app} decompositions of moduli spaces of dimension-zero and-one open-closed hyperbolic vertices are summarized and in appendix~\ref{appB} the off-shell data for disks with two bulk punctures are characterized when open strings are taken on-shell.

\section{Preliminaries} \label{sec:first}

In this section we summarize the fundamentals of the geometric master equation and the construction by Cho~\cite{Cho:2019anu} in order to establish our conventions. We also briefly discuss the local coordinates for open-closed hyperbolic string vertices. General introduction to the geometric master equation can be found in~\cite{Zwiebach:1992ie} and reader can refer to~\cite{buser2010geometry} for mathematically-oriented introduction to hyperbolic geometry and the collar lemma. The instructions on how to derive local coordinates for hyperbolic vertices are given in~\cite{Firat:2023glo}.

\subsection{Geometric master equation}

We begin by reviewing the geometric master equation relevant for the open-closed SFT~\cite{Zwiebach:1992ie, Zwiebach:1997fe}. We define $\mathcal{M}^{g,n}_{b, \{m_i\}}$ to be the moduli space of Riemann surfaces of genus $g$ with $n$ bulk punctures (corresponding to closed strings) and $m_i$ boundary punctures (corresponding to open strings) at the $i$-th boundary component. Here $b$ is the number of open string boundaries so $i=1, \cdots, b$. For convenience we sometimes use $m= \sum_i m_i$, which would be apparent from the context. Note that $m_i$ can be zero, i.e. there can be open string boundaries carrying no boundary punctures.

Off-shell amplitudes in SFT require selecting local coordinates around each puncture. These are described by a collection of maps $f_i (w)$ from (half-)unit disks $0 < |w| \leq 1$ to a Riemann surface for closed (open) strings. The overall phase of the coordinates around bulk punctures is irrelevant.\footnote{This is due to off-shell level-matching of closed strings. See~\cite{Okawa:2022mos,Erbin:2022cyb} for recent attempts to relax this constraint.}. We take the coordinate patches $z=f_i(w)$ on surfaces to be non-overlapping---except possibly at their boundaries. All possible choices of local coordinates around punctures with these constraints define the bundle $\widehat{\mathcal{P}}^{g,n}_{b, \{m_i\}} \to \mathcal{M}^{g,n}_{b, \{m_i\}}$. String vertices $\mathcal{V}$ are then defined to be the following formal sum of singular chains $\mathcal{V}^{g,n}_{b, \{m_i\}} $ of this bundle~\cite{Zwiebach:1997fe}:
\begin{align} \label{eq:SV}
\mathcal{V}
= \sum_{g, n, b, \{m_i\}} \hbar^p \, \kappa^q \,
\mathcal{V}^{g,n}_{b, \{m_i\}} \, .
\end{align}
The sum runs over all non-negative values of $g,n,b,m_i$ with the restriction
\begin{subequations}
	\begin{align}
	&n \geq 3 \quad \text{for} \quad g=b=0 \, , \\
	&n \geq 1 \quad \text{for} \quad g=1, \, b=0 \, , \\
	&m_1 \geq 3 \quad \text{for} \quad g=0, \, b=1 \, .
	\end{align}
\end{subequations}
The parameter $\hbar$ and open string coupling constant $\kappa$ are formal variables and their powers $p$ and $q$ in~\eqref{eq:SV} are given by
\begin{align}
p = 2g + {n \over 2} + b - 1, \quad \quad
q = 4g + 2n + 2b + m - 4 
\, .
\end{align}

In order to consistently implement the gauge invariance of SFT string vertices $\mathcal{V}$ are demanded to satisfy~\textit{the geometric master equation}~\cite{Zwiebach:1992ie,Zwiebach:1997fe}
\begin{align} \label{eq:BV}
\partial \mathcal{V} =
- {1 \over 2} \{\mathcal{V}, \mathcal{V}\} 
- \hbar \Delta \mathcal{V}  \, ,
\end{align}
which encodes a homological recursion relation for $\mathcal{V}^{g,n}_{b, \{m_i\}}$. Here $\partial$ is the boundary operator for chains on the bundle $\widehat{\mathcal{P}}^{g,n}_{b, \{m_i\}} \to \mathcal{M}^{g,n}_{b, \{m_i\}}$ and operations $\{\cdot, \cdot \}$ and $\Delta$ are given by
\begin{subequations}
\begin{align}
&\{\cdot, \cdot \}  = \{\cdot, \cdot \}_c + \{\cdot, \cdot \}_o \, , \\
&\Delta = \Delta_c + \Delta_o= \Delta_c + \Delta_o^{(1)} + \Delta_o^{(2)} \, ,
\end{align}
\end{subequations}
where $\{\cdot, \cdot \}_{c,o}$ is the operation of gluing bulk (boundary) punctures on disjoint surfaces by the plumbing fixture and symmetrization of the labels of remaining punctures. Recall that the plumbing fixture for the local coordinates $w_1, w_2$ around two bulk or boundary punctures are given by
\begin{align} \label{eq:plumbing}
	w_1 w_2 = e^{2 \pi i \theta / L_c}, \quad \text{and} \quad
	w_1 w_2 = - 1 \, ,
\end{align}
respectively. We have to twist-sew the local coordinates around bulk punctures, which means we have to consider all possible surfaces resulting from taking $0\leq \theta < L_c$. Similarly, the operation $\Delta_c$ denotes twist-sewing a pair of bulk punctures on the same surface, $\Delta_o^{(1)} $ denotes sewing a pair of boundary punctures on the same boundary, and $\Delta_o^{(2)} $ denotes sewing a pair of boundary punctures on the same surface \textit{but} on different boundaries. Making the distinction between $\Delta_o^{(1)} $ and $\Delta_o^{(2)} $ is going to be important for the large $N$ limit.

\subsection{Open-closed hyperbolic string vertices}

We summarize the key points of the construction of~\cite{Cho:2019anu} in this subsection. First, we define the set\footnote{This definition excludes the annulus and single bulk-punctured disk. We comment on them below.}
\begin{align} \label{eq:Cho}
\widetilde{\mathcal{V}}^{g,n}_{b, \{m_i\}} (L_o, L_c) = \left\{\Sigma \in 
\mathcal{M}^{g,n}_{b, \{m_i\}}(L_o, L_c) \; \big| \; \text{sys}(\Sigma) \geq L_c  \quad \text{and} \quad
\text{psys}(\Sigma) \geq L_o
\right\} \, ,
\end{align}
where $\mathcal{M}^{g,n}_{b, \{m_i\}}(L_o, L_c)$ is the moduli space of Riemann surfaces (endowed with hyperbolic metric $K=-1$) of genus $g$ with $n$ geodesic boundaries of length $L_c$ (called \textit{$c$-borders}) and $b$~\textit{piecewise}-geodesic boundaries with $2m_i$ pieces for $i=1, \cdots, b$ whose alternating piecewise components have length $L_o$. The sides of length $L_o$ are~\textit{$p$-sides} and the remaining sides are~\textit{$b$-sides}. When $m_i = 0$ the boundary is called a~\textit{$b$-border}, which we consider a special kind of $b$-side.

The systole $\text{sys}(\Sigma)$ is defined as the length of the shortest non--contractible closed geodesic on the surface $\Sigma$ non-homotopic to a $c$-border. Notice the length of the $b$-border can affect the value of the systole by our definition. Defining $p$-systole $\text{psys}(\Sigma)$ precisely, on the other hand, first requires introducing the notion of~\textit{$p$-geodesic}. A $p$-geodesic is a simple nontrivial open geodesic of a Riemann surface whose end points are on the $b$-sides and/or $b$-borders and it is the shortest curve among its homotopy class where gliding along $b$-sides are allowed. Then the $p$-systole $\text{psys}(\Sigma) $ is defined as the length of the shortest $p$-geodesic of the surface $\Sigma$ non-homotopic to a $p$-side, in accord with the definition of the ordinary systole.

The string vertices $\mathcal{V}^{g,n}_{b, \{m_i\}} (L_o, L_c)$ are constructed via grafting semi-infinite flat cylinders of circumference $L_c$ to $c$-borders and grafting semi-infinite flat strips of side length $L_o$ to $p$-sides to surfaces in $\widetilde{\mathcal{V}}^{g,n}_{b, \{m_i\}} (L_o, L_c) $
\begin{subequations} \label{eq:HyperbolicVertices}
\begin{align} 
\mathcal{V}^{g,n}_{b, \{m_i\}}  (L_o, L_c) = \text{gr}' \left(\widetilde{\mathcal{V}}^{g,n}_{b, \{m_i\}} \right)  \, ,
\end{align}
where $ \text{gr}'$ stands for this grafting operation. Resulting surfaces have $n$ bulk punctures, $m_i$ boundary punctures on $b$ boundary components. These are further endowed with
\begin{align} \label{eq:Rest}
&\widetilde{\mathcal{V}}^{0,1}_{1, \{m_1 = 0\}} (L_c) =
\left\{
\text{A flat circle of circumference $L_c$}
\right\} \, ,
&\widetilde{\mathcal{V}}^{0,0}_{2, \{m_1 = 0,  m_2 = 0\}} = \emptyset \, ,
\end{align}
\end{subequations}
and the actual vertices $\mathcal{V}^{0,1}_{1, \{m_1 = 0\}} (L_c) $ and $\mathcal{V}^{0,0}_{2, \{m_1 = 0,  m_2 = 0\}}$ are given by grafting a semi-infinite flat cylinder. The latter one is still empty. The vertices $\mathcal{V}^{g,n}_{b, \{m_i\}}  (L_o, L_c)$ are non-empty in general~\cite{Costello:2019fuh,Cho:2019anu}.

The novelty of Cho's construction was identifying the collar lemma for the relevant Riemann surfaces to open-closed vertices, which eventually placed the following conditions on $L_c$ and $L_o$
\begin{subequations}  \label{eq:Condition}
\begin{align}
0 < L_c \leq L_\ast  \quad \quad \text{and} \quad \quad
\sinh L_c \, \sinh L_o \leq 1 \, , 
\end{align}
where~\textit{the critical length $L_\ast$} solves
\begin{align}
\sinh {L_\ast \over 2} \sinh L_\ast = 1 \implies
L_\ast \approx 1.21876 \, .
\end{align}
\end{subequations}
Upon assuming these constraints, the string vertices~\eqref{eq:HyperbolicVertices} can be shown to solve the geometric master equation~\eqref{eq:BV} using this generalized collar lemma~\cite{Cho:2019anu}. 

Hyperbolic vertices of Cho is perfectly suited to construct any open-closed SFT, but it may not be the most optimal parametrization as far as the large $N$ limit~\textit{and} the physics of closed strings are concerned. As we pointed out earlier, polyhedral vertices appear to better suited to classical theories than arbitrary choices---including generic hyperbolic ones. Polyhedral closed string vertices correspond to the $L_c \to \infty$ limit of closed hyperbolic string vertices~\cite{Costello:2019fuh,Firat:2023glo} and this is essentially the primary reason behind their ``simplicity''. This limit forces external strings to overlap, i.e. interactions to take place on a polyhedron. As a result the coefficients appearing in the SFT action gets minimized, leading an efficient parametrization. See the discussion in~\cite{Belopolsky:1994sk,upcoming_work}.

It would be desirable to find a solution to the geometric master equation where an analogous limit exists for the planar limit given its semi-classical nature. So our objective should be relaxing the conditions~\eqref{eq:Condition} on hyperbolic string vertices to allow for a limit where $L_c$ can be taken arbitrarily large. We are going to tackle this issue in the next section.

\subsection{A brief remark on the local coordinates}

It is beneficial to briefly comment on the local coordinates of open-closed hyperbolic string vertices at this point, which will be useful later. This subsection is aimed to provide a ``recipe'' for deriving them, so we mostly emphasize novel ingredients due to having open string boundaries and boundary punctures. For more details on this recipe and the actual computations refer to~\cite{Firat:2021ukc,Firat:2023glo,Firat:2023suh, Hadasz:2006vs}.

We begin, as in closed strings, with the Fuchsian equation\footnote{The following logic applies to all open-closed hyperbolic vertices, except for $\mathcal{V}^{0,1}_{1, 0}(L_c)$. The local coordinates for this vertex is given by the Cayley map~\eqref{eq:Cayley} as we shall see in section~\ref{sec:MRV}.}
\begin{align} \label{eq:Fuchsian}
	\partial^2 \psi + {1 \over 2} T(z) \psi = 0 \, ,
\end{align}
on an arbitrary coordinate patch $z$ of a punctured Riemann surface with open string boundaries $\Sigma$. The objects $\psi(z)$ and $T(z)$ transform as
\begin{align} \label{eq:Transform}
\psi(z) = \left({\partial  \widetilde{z} \over \partial z}\right)^{-1/2} \widetilde{\psi} ( \widetilde{z}) \, ,
\quad \quad
T(z) = \left({\partial  \widetilde{z} \over \partial z}\right)^{2} \widetilde{T} ( \widetilde{z}) + \{\widetilde{z}, z\} \, ,
\end{align}
under the holomorphic transition maps $z \to \widetilde{z}(z)$, so the equation~\eqref{eq:Fuchsian} is coordinate-independent. Here $\{\cdot, \cdot\}$ is the Schwarzian derivative
\begin{align}
	\{\widetilde{z}, z\}  = {\p^3 \widetilde{z} \over \p \widetilde{z}} -
	{3 \over 2} \left( {\p^2 \widetilde{z} \over \p \widetilde{z}} \right)^2 \, .
\end{align}

Next, we demand ``stress-energy tensor'' $T(z)$ in~\eqref{eq:Fuchsian} is chosen such that the solutions $\psi(z)$ can have hyperbolic $PSL(2, \mathbb{R})$ monodromy around all punctures.\footnote{The definition of monodromy around boundary punctures is delicate. We will address this below.} We additionally demand
\begin{align} \label{eq:Cond}
	\overline{T(z)} = T(\overline{z}) \, ,
\end{align}
on any coordinate patch $z$ that contains a part of the open string boundary $R$. We take $\mathrm{Im} z \geq 0$ for these patches. This additional condition~\eqref{eq:Cond} is required to guarantee the existence of a normalized basis of solutions $\psi^{\pm}_R(z)$ to~\eqref{eq:Fuchsian} that satisfies
\begin{align} \label{eq:Normalization}
	\Psi_R(z) \equiv \begin{bmatrix}
	\psi^{-}_R(z) \\ \psi^{+}_R(z)
	\end{bmatrix} \, ,\quad \quad
	\overline{\Psi_R(z)} = \Psi_R(\overline{z}) \, ,\quad \quad
	\partial \Psi^T_R(z) \sigma_2 \Psi_R(z) = i \, ,
\end{align}
where $\sigma_2$ is the second Pauli matrix. As we shall see shortly, the existence of $\Psi_R(z)$ is necessary to construct open string boundaries in the geometry. But for now just notice the solutions $\psi^{\pm}_R(z)$ are real and regular along the open string boundary $z = \overline{z}$. We relate any other normalized basis of solutions to~\eqref{eq:Fuchsian} $\Psi(z) $ to the basis $\Psi_R(z) $ with a complex matrix $S$
\begin{align} \label{eq:Rel}
	\Psi(z) = S \Psi_R (z) \, , \quad \quad
	S \in PSL(2, \mathbb{C}) \, .
\end{align}

Now, recall that a hyperbolic metric (possibly with singularities) $ds^2 = e^{\varphi} |dz|^2$ can be related to the solutions to the Fuchsian equation by the equality~\cite{Firat:2021ukc}
\begin{align} \label{eq:Soln}
	e^{-\varphi / 2} = \pm {1 \over 2} \Psi^T(z) \sigma_2 \overline{\Psi(z)} \, ,
\end{align}
upon setting $T(z) = - (\partial \varphi)^2/2 + \partial^2 \varphi$. Realizing hyperbolic monodromies around punctures simultaneously gives rise to the hyperbolic metric on~\textit{$c$-bordered} surfaces as shown in~\cite{Firat:2021ukc} which can be used to derive the local coordinates around bulk punctures. However, this isn't sufficient for open-closed vertices: we must supplement the condition of $ds^2$ being a hyperbolic metric with boundary conditions at open string boundary $R$. Since the construction of~\cite{Cho:2019anu} takes them to be geodesics of the metric $ds^2$, the boundary condition at $R$ is given by
\begin{align} \label{eq:BC}
	{1 \over 2}e^{-\varphi/2 }n^a \partial_a \varphi \big|_R = 0 \quad \overset{\text{locally}}{\implies} \quad
	( \partial - \overline{\partial}) e^{-\varphi/2} \big|_R  = 0 \, .
\end{align}
Here $n^a$ is the normal vector to $R$. Upon using~\eqref{eq:Rel},~\eqref{eq:Soln}, the boundary condition~\eqref{eq:BC} becomes
\begin{align}
	0 = ( \partial - \overline{\partial}) \Psi^T_R(z) S^T \sigma_2 \overline{S} \Psi_R(  \overline{z} ) \big|_R
	= \left[ \partial \Psi^T_R(z) S^T \sigma_2 \overline{S} \Psi_R(  \overline{z} ) 
	- \Psi^T_R(z) S^T \sigma_2 \overline{S}  \overline{\partial} \Psi_R(  \overline{z} ) \right]_R \, .
\end{align}
This is a scalar expression, so we can take its trace over the solution basis and use its cyclicity to rewrite the boundary condition as
\begin{align} \label{eq:tr}
	0 = \tr \left[ \left(\Psi_R(  \overline{z} ) \partial \Psi^T_R(z) -\overline{\partial} \Psi_R(  \overline{z}) \Psi^T_R(z)\right)  S^T \sigma_2 \overline{S}
	\right]_R = i \, \tr \left[ \sigma_2 S^T \sigma_2 \overline{S} \right] \, .
\end{align}
In the second line we have used the normalization~\eqref{eq:Normalization} of the solution $\Psi_R$ evaluated at the open string boundary $R$. We see that we have to impose $\tr ( \sigma_2 S^T \sigma_2 \overline{S} ) = 0$ in order to make sure open string boundaries are geodesics of the hyperbolic metric. It shouldn't be too surprising that imposing boundary conditions also lead to monodromy-like conditions. We emphasize that it was crucial to impose $\overline{T(z)} = T(\overline{z})$, otherwise the geodesic boundary conditions cannot be imposed generally. Moreover, we see the basis $\Psi_R$~\textit{cannot} satisfy the extra trace condition~\eqref{eq:tr}, as $S=1$ results in $\tr ( \sigma_2 S^T \sigma_2 \overline{S} ) = 2$. This makes sense, taking $\Psi = \Psi_R$ in~\eqref{eq:Soln} leads to a singular metric by~\eqref{eq:Soln}.

We are now ready to introduce the notion of monodromy around boundary punctures. We do this through ``doubling trick''. That is, we extend the patch containing a boundary puncture to $\mathrm{Im}(z) \leq 0$ from $\mathrm{Im}(z)\geq0$ and define the associated $T(z)$ on the lower-half plane via $T(z) = \overline{T(\overline{z})}$, which makes it holomorphic for all $z$. The monodromy around a boundary puncture is then given by the ordinary monodromy of this puncture in the extended patch. Recall hyperbolic $PSL(2,\mathbb{R})$ monodromy introduces double poles in $T(z)$ at bulk punctures whose residue is related to the length of the $c$-borders~\cite{Firat:2021ukc}. The definition above implies there are also double poles at the positions of boundary punctures. Their residues are related to the length of the $p$-sides.

The problem described above is the hyperbolic monodromy problem in the context of open-closed hyperbolic string vertices. The solutions $\Psi(z)$ that realize hyperbolic monodromy can subsequently be used to derive explicit expressions for the local coordinates around any puncture. Note that this monodromy problem is expected to have a unique solution for a given surface by the uniformization theorem~\cite{Firat:2023glo}. Similar to closed hyperbolic string vertices, finding $T(z)$ that leads to correct monodromies involves specifying $\dim \mathcal{M}^{g,n}_{b,{m_i}} = 6g -6 +2n + 3b + m $ real~\textit{accessory parameters} and this is the non-trivial part of the procedure. The counting for the accessory parameters can be made similar to closed string counterparts~\cite{Firat:2023suh}.

Let us illustrate the additional complexity open strings introduce more explicitly in the simplest situation: the disk with $n$ bulk and $m$ boundary punctures. We describe this geometry on the upper-half plane (UHP) $z$ with $\mathrm{Im} z \geq 0$ by placing the bulk punctures at $z=\xi_i \in \mathbb{C}$ for $i=1,\cdots, n$ and boundary punctures at  $z=x_j \in \mathbb{R}$ for $j=1, \cdots, m$. The most general expression for $T(z)$ that may solve the monodromy problem is then given by
\begin{align} \label{eq:DiskT}
	T(z; \mathcal{V}^{0,n}_{1,m}(L_{o,j}, L_{c,i})) = \sum_{i=1}^n \left[ {\Delta_i \over (z-\xi_i)^2} + {\Delta_i \over (z-\overline{\xi_i})^2} + {c_i \over z- \xi_i} + {\overline{c_i} \over z - \overline{\xi_i}} \right]
	+ \sum_{j=1}^m \left[ {\Delta^R_j \over (z- x_j)^2} + {c_j^R \over z-x_j} \right] \, ,
\end{align}
according to our discussion above. The condition~\eqref{eq:Cond} is satisfied upon taking $\Delta_i, \Delta_j^R, c_j^R \in \mathbb{R}$. Here $\Delta_i$ and $\Delta_j^R$ are related to the $c$-borders/$p$-sides lengths $L_{c,i}, L_{o,j}$ around punctures via
\begin{align}
	\Delta_i = {1 \over 2} + {1 \over 2} \left({L_{c,i} \over 2 \pi}\right)^2 \, , \quad \quad
	\Delta^R_j = {1 \over 2} + {1 \over 2} \left({L_{o,j} \over \pi}\right)^2  \, .
\end{align}
Even though we consider distinct $c$-border/$p$-side lengths above, we only need to consider the situation $L_{c,i} = L_c$ and $L_{o,j} = L_o$ for SFT purposes. We suppress the subscript on $\Delta, \Delta^R$ when $c$-border/$p$-side lengths are taken equal. 

Here $c_i \in \mathbb{C}, c_j^R \in \mathbb{R}$ are the accessory parameters of this situation. Since $z=\infty$ has a trivial monodromy there are three linear constraints $c_i, c_j^R$ satisfy. These are
\begin{subequations} \label{eq:Cons}
\begin{align}
	&\sum_{i=1}^n \left(c_i + \overline{c_i}\right) + \sum_{j=1}^m c_j^R = 0 \, , \\
	&\sum_{i=1}^n \left(2 \Delta_i +  c_i \xi_i + \overline{ c_i \xi_i } \right)
	+ \sum_{j=1}^m \left(\Delta^R_j + c_j^R x_j  \right) = 0 \, , \\ 
	&\sum_{i=1}^n \left(2 \Delta_i \xi_i  + 2 \Delta_i \overline{\xi_i}  +  c_i \xi_i^2 + \overline{c_i \xi_i^2}\right)
	+ \sum_{j=1}^m \left(2 \Delta_j^R x_j + c_j^R x_j^2 \right) = 0 \, .
\end{align}
\end{subequations}
These constraints leave us with $2n + m - 3$ real accessory parameter that needs to be fixed for $\mathcal{V}^{0,n}_{1,m}(L_{o,j}, L_{c,i})$, consistent with the counting above.

Observe that there are no undetermined accessory parameters for the cases $n=0, m=3$ and $n=1, m=1$ and the associated stress-energy tensors for them is simply given by
\begin{subequations} \label{eq:WKB}
\begin{align}
	&T(z; \mathcal{V}^{0,0}_{1,3}(L_{o,1}, L_{o,2}, L_{o,3})) = {\Delta_1^R \over z^2} + {\Delta_2^R \over (z-1)^2} + {\Delta_1^R  + \Delta_2^R  - \Delta_3^R  \over z} + {-\Delta_1^R  -\Delta_2^R  + \Delta_3^R  \over z-1} \, ,\\
	&T(z; \mathcal{V}^{0,1}_{1,1}(L_o, L_c)) = {\Delta \over (z-i)^2} + {\Delta \over (z+i)^2} + {i \, (2\Delta + \Delta^R) \over 2 \, (z-i)}
	+ {-i \, (2\Delta + \Delta^R) \over 2 \, (z+i)} + {\Delta^R \over z^2} \, ,
\end{align}
\end{subequations}
upon implementing constraints~\eqref{eq:Cons}. Here we have used $PSL(2,\mathbb{R})$ freedom on the UHP $z$ to place the punctures at $z=0,1,\infty$ and $z=0,i$ for the $n=0,m=3$ and $n=1, m=1$ cases respectively. 

Except for the cases above, the rest of the open-closed hyperbolic string vertices involves undetermined accessory parameters. Following the ideas of~\cite{Firat:2023glo, Firat:2023suh,Hadasz:2006vs}, we may attempt to write a Polyakov conjecture that would fix them in terms of derivatives of a suitably regularized on-shell Liouville action with FZZT boundary conditions and subsequently construct them using classical conformal bootstrap. In practice, however, this isn't really necessary for Riemann surfaces with open string boundaries as long as one bootstraps the local coordinates of closed Riemann surfaces first. We can use the latter to construct the local coordinates of the former by appropriately ``cutting'' the surfaces along their geodesics. 

\section{String vertices for the large $N$ limit} \label{sec:N}

In this section we consider string vertices for the large $N$ limit. As we shall see, it is sufficient to consider genus $0$ punctured Riemann surfaces with open string boundaries in accord with~\cite{Maccaferri:2023gof}. Subsequently we argue that they obey a version of geometric master equation and show that hyperbolic geometry can be used to solve this equation with a milder constraint relative to~\cite{Cho:2019anu}. We consider limits of these vertices where closed or open strings are taken on-shell and investigate the consequences. 

\subsection{Geometric planar master equation}

In this subsection we derive the string vertices relevant for the open-closed SFT in the large $N$ limit. We begin with the formal sum
\begin{align} \label{eq:FormalVar}
\mathcal{V}
= \sum_{g, n, b, \{m_i\}} \hbar^p \, \kappa^q \, N^b \,
\mathcal{V}^{g,n}_{b, \{m_i\}} \, .
\end{align}
Here we explicitly stripped the number of ``colors'' $N$ resulting from Chan-Paton factors of open strings ending on $N$ identical D-branes from amplitudes and declared it as formal variable. This requires normalizing the trace $\tr$ over Chan-Paton factors. That is, we use
\begin{align}
	\tr(\cdots) \to \tr'(\cdots) \equiv {1 \over N} \tr(\cdots) \, ,
\end{align}
in the SFT interactions resulting from $\mathcal{V}^{g,n}_{b, \{m_i\}} $ now. It is apparent that the normalized trace $\tr'$ stays finite as $N \to \infty$ and the divergence associated with having large number of branes resides in the formal variable in~\eqref{eq:FormalVar}. It is worth mentioning that we take trace over Chan-Paton factors for each boundary separately, which justifies the power of $N$ in~\eqref{eq:FormalVar}

Let us make some further arrangement to the way we write~\eqref{eq:FormalVar}:
\begin{align} \label{eq:See}
\mathcal{V} = 
\hbar
\sum_{g, n, b, \{m_i\}}
(\hbar \kappa^2 N)^{2g + b - 2} N^{2-2g}
\left(  \hbar^{n/2} \kappa^{2n+m} \mathcal{V}^{g,n}_{b, \{m_i\}} \right) \, .
\end{align}
There are few things to observe. First, we can rescale the open string coupling constant $\kappa \to \kappa \hbar^{-1/2}$ and absorb the remaining $\hbar$ dependence to vertices $\mathcal{V}^{g,n}_{b, \{m_i\}}$ by rescaling the string fields. There is an overall $\hbar$, which we can simply drop out using $S/\hbar$ instead of $S$ in SFT, where $S$ is the Batalin-Vilkovisky (BV) action determined by $\mathcal{V}$. This whole operation effectively sets $\hbar =1$. Moreover, we can further scale string fields to absorb $\kappa$ into vertices $\mathcal{V}^{g,n}_{b, \{m_i\}}$, leaving us with the combination $\kappa^2 N$ and $N$ itself. We keep using the notation $\mathcal{V}^{g,n}_{b, \{m_i\}}$ for string vertices even after these scalings.

Now we can consider the large $N$ limit. For this, we define the `t Hooft coupling
\begin{align}
\lambda \equiv \kappa^2 N \, ,
\end{align}
which stays fixed as $N \to \infty$, $\kappa \to 0$. Then~\textit{the planar string vertices} $\mathcal{V}_{pl}$ are given by
\begin{align} \label{eq:PlanarHyperbolicVertices}
\mathcal{V}_{pl} \equiv \lim_{N \to \infty} {\mathcal{V} \over N^2} 
= \sum_{n, b, \{m_i\}}  \lambda^{b -2} \mathcal{V}^{0,n}_{b, \{m_i\}} \, .
\end{align}
We point out only genus $0$ surfaces survive this limit. This result is consistent with the action given in~\cite{Maccaferri:2023gof}.\footnote{Notice we use $\kappa$ for the~\textit{open} string coupling constant, not for the~\textit{closed} string coupling constant like in~\cite{Maccaferri:2023gof}.} We are going to denote these surfaces by~\textit{planar Riemann surfaces} for convenience. Observe that we still work with small string coupling constant in the large $N$ limit so the worldsheet keeps making sense whenever it is unambiguous.

The proper implementation of gauge symmetries in SFT from~\eqref{eq:BV} demands planar string vertices to satisfy the equation ($\hbar = 1$)
\begin{align} \label{eq:PlanarBV}
\p \mathcal{V}_{pl} + {1 \over 2} \{\mathcal{V}_{pl}, \mathcal{V}_{pl} \} + \Delta_o^{(1)}  \mathcal{V}_{pl}  = 0\, ,
\end{align}
which we shall call~\textit{the geometric planar master equation}. This equation describes a recursion relation among the constituents of $\mathcal{V}_{pl}$. This recursion relation is schematically shown in figure~\ref{fig:planar_bv}.
\begin{figure}
	\hspace{-0.3in}
	\includegraphics[width=1.1\textwidth,height=0.35\textwidth]{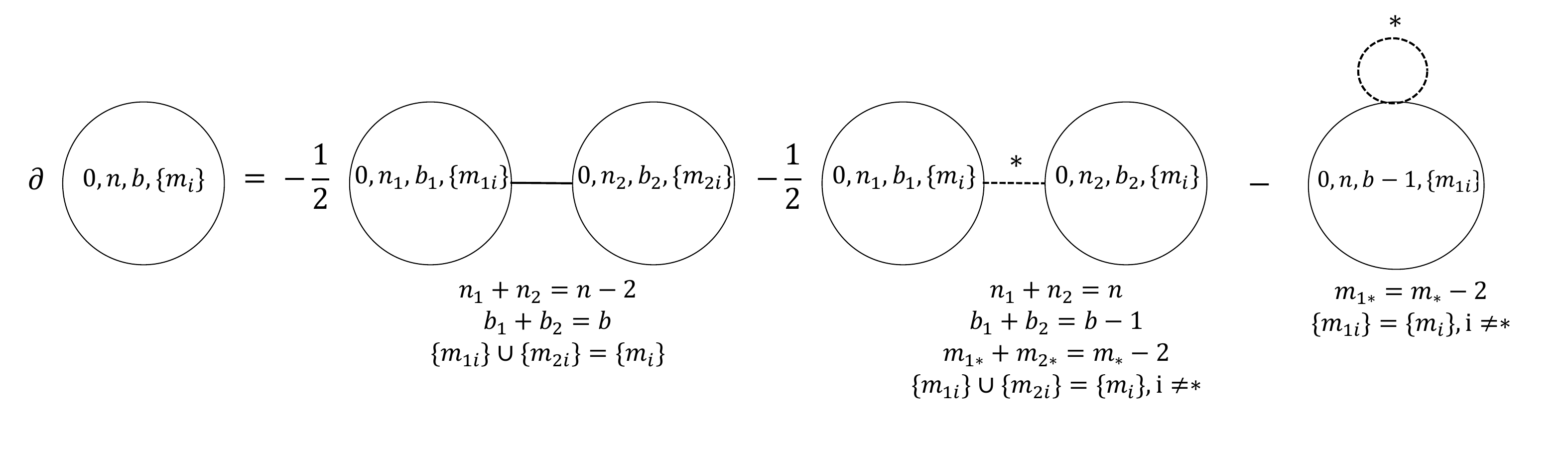}
	\caption{\label{fig:planar_bv}The schematic representation of the recursion relation defined by~\eqref{eq:PlanarBV}. Each blob stands for a string vertex and the solid (dashed) line denotes a collapsed closed (open) string propagator. The star on the dashed line stands for the particular boundary that sewed open string punctures belong. We implicitly sum over all possible vertices that satisfy the relations below the blobs.}
\end{figure}

Compared to the full geometric master equation~\eqref{eq:BV},  $\Delta_c$ and $\Delta_o^{(2)}$ are missing. This can be understood as follows. The operation corresponding to $\Delta_c$ increases the genus by 1, which is suppressed in the large $N$ limit. Other way to put\footnote{Note that the sewing operations themselves carry formal variables after rearrangement~\eqref{eq:See}.}
\begin{align} \label{eq:Dc}
\Delta_c \sim {\lambda^2 \over N^2} \, .
\end{align}
Similarly, the operation corresponding to $\Delta_o^{(2)}$ decreases the number of boundaries by 1 but increases the genus by 1, which gets suppressed in the large $N$ limit as
\begin{align} \label{eq:Do2}
\Delta_o^{(2)} \sim {\lambda \over N^2} \, .
\end{align}
The rest of the operations are clearly $\mathcal{O}(N^0)$ since they don't increase the genus and so they are still present. More precisely, their scalings are given by
\begin{align}
	\{\cdot, \cdot\}_o \sim \lambda \, , \quad \quad
	\{\cdot, \cdot\}_c \sim \lambda^2 \, , \quad \quad
	\Delta_o^{(1)} \sim \lambda \, .
\end{align}
These results are consistent with~\cite{Maccaferri:2023gof}. We point out the geometric planar master equation~\eqref{eq:PlanarBV} has been observed to be a sub-recursion relation of the full geometric master equation~\eqref{eq:BV} in~\cite{Zwiebach:1997fe} long time ago, however its connection to large $N$ limit has gone unnoticed. 

It is beneficial to briefly remark on how $1/N$ corrections work in this framework. As it can be seen from~\eqref{eq:See}, string vertices describing $1/N^{2g}$ corrections are given by including genus $g$ surfaces into consideration. This is consistent with the notion that large $N$ limit is semi-classical and the associated corrections are effectively quantum. We point out solution to~\eqref{eq:PlanarBV} including $1/N$ corrections is only possible if we include surfaces with arbitrary number of genus: corrections truncated to some order does not define a sub-recursion like in the strictly planar limit. That is,
\begin{align}
\mathcal{V}_{pl}^{1PI} 
= \sum_{n, b, \{m_i\}}  \lambda^{b -2} \sum_{g=0}^\infty N^{-2g} \mathcal{V}^{g,n}_{b, \{m_i\}} \, ,
\end{align}
solves~\eqref{eq:PlanarBV} as well. Notice doing this resummation leads to ``1PI action'' for the large $N$ limit in a certain sense, following the logic in~\cite{Sen:2014dqa}. In this theory there are still no non-separating open or closed string degenerations---they are described within vertex and the remaining Feynman regions. In any case, the full open-closed geometric master equation~\eqref{eq:BV} can be considered (and solved) order-by-order in $1/N^2$ with~\eqref{eq:Dc} and~\eqref{eq:Do2}.

\subsection{Planar hyperbolic string vertices}

We use hyperbolic geometry to solve the geometric planar master equation~\eqref{eq:PlanarBV} in this subsection. Our guiding principle is to eliminate the use of collar lemma in the arguments when appropriate, as this was the main culprit behind the placing constraints on $L_c$~\eqref{eq:Condition} for open-closed hyperbolic vertices. Our main result is the following:
\begin{claim}
	The string vertices
	\begin{align} \label{eq:StringVertices}
	\mathcal{V}_{pl} (L_o, L_c)
	= \sum_{n, b, \{m_i\}}  \lambda^{b -2} \mathcal{V}^{0,n}_{b, \{m_i\}}(L_o, L_c) \, ,
	\end{align}
	where $\mathcal{V}^{0,n}_{b, \{m_i\}}  (L_o, L_c) $ defined in~\eqref{eq:HyperbolicVertices}, solve the geometric planar master equation~\eqref{eq:PlanarBV} as long as the border lengths $L_c, L_o$ satisfy the constraint
	\begin{align} \label{eq:Constraint}
		\sinh {L_c} \, \sinh {L_o} \leq 1 \, ,
	\end{align}
	see figure~\ref{fig:Constraint}. We call these vertices planar hyperbolic string vertices.
\end{claim}
\begin{proof}
Almost all of the proof proceeds as in~\cite{Cho:2019anu}. We just notice that the constraint $L_c \leq L_\ast$ is coming from the proof of the containment of $\{\mathcal{V}, \mathcal{V}\}_c$ and $\Delta_c \mathcal{V}$ in $\partial \mathcal{V}$. The latter is irrelevant for us, while the former can be modified using the argument in~\cite{Costello:2019fuh} as follows. We want to make sure that when $c$-borders of two disjoint surfaces $\Sigma_1, \Sigma_2 \in  \mathcal{V}_{pl}(L_o,L_c)$ are twist-sewed we don't create any non-contractible closed geodesic shorter than $L_c$ in the combined surface $\Sigma_1 \sharp \Sigma_2$. We can't create any $p$-geodesic shorter than $L_o$ already, given that~\eqref{eq:Constraint} is satisfied, see~\cite{Cho:2019anu}. Since $\text{sys}(\Sigma_1), \text{sys}(\Sigma_1) \geq L_c$, we only need to worry about non-contractible closed geodesics shorter than $L_c$ that cross~\textit{both} surfaces $\Sigma_1, \Sigma_2$.

So, denote the sewed $c$-border by $\gamma$, which is a simple closed geodesic with length $L_c$, and suppose there is a simple closed geodesic $\delta$ with $\delta \subset \Sigma_1 \sharp \Sigma_2$ that transverses $\Sigma_1 \sharp \Sigma_2$ whose length is shorter than $L_c$ for the sake of contraction, i.e. we have
\begin{align}
	\gamma = L_c \quad \quad
	\text{and} \quad \quad
	\delta < L_c \, ,
\end{align}
Here we used the same letters for geodesics and their lengths. Note that it is sufficient to consider a~\textit{simple} geodesic $\delta$: if it were self-intersecting we can always create shorter non-intersecting curves by performing the operation shown in figure~\ref{fig:proof} at intersections and consider a simple shorter geodesic homotopic to one of the resulting simple curves.
\begin{figure}
	\centering
	\includegraphics[width=\textwidth,height=5cm]{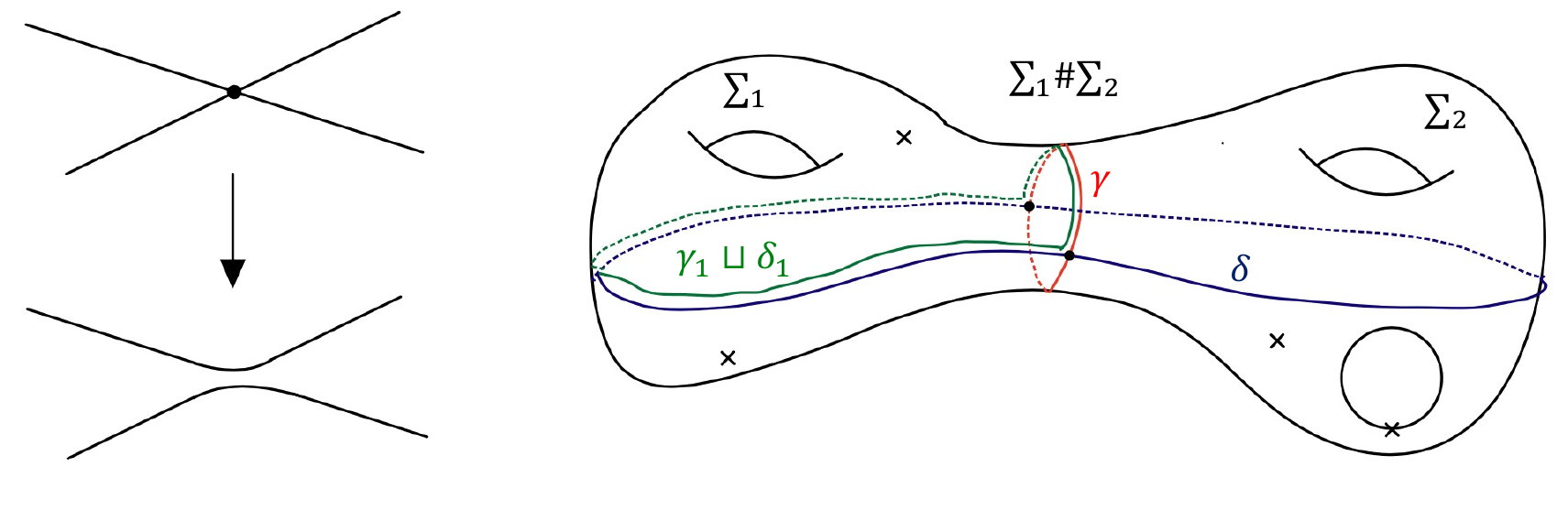}
	\caption{\label{fig:proof}The operation eliminating self-intersections~(\textit{right}) and a pictorial representation of the configuration in the proof~(\textit{left}).}
\end{figure}

Observe that $\gamma$ divides $\delta$ by two and vice-versa, see figure~\ref{fig:proof}. We denote these two components by $\gamma_1, \gamma_2$ and $\delta_1, \delta_2$ respectively, so we have
\begin{align}
	\gamma = \gamma_1 \sqcup \gamma_2 \quad \quad
	\text{and} \quad \quad
	\delta = \delta_1 \sqcup \delta_2 \, .
\end{align}
One of the components of both $\gamma$ and $\delta$ has to be smaller than $L_c/2$. Say
\begin{align}
	\gamma_1 \leq {L_c \over 2} \quad \quad
	\text{and} \quad \quad
	\delta_1 < {L_c \over 2} \, ,
\end{align}
without loss of generality. This implies the length of the curve $\gamma_1 \sqcup \delta_1$  has to be smaller than $L_c$. But this curve can be ``slided'' to one of the surfaces, which leads to a contradiction with the assumption $\Sigma_1, \Sigma_2 \in  \mathcal{V}_{pl}(L_o,L_c)$ as it gives rise to a non-contractible simple closed geodesic whose length is smaller than the systole of the surface, $L_c$. This sliding produces a~\textit{non-contractible} simple closed geodesic: if this weren't the case $\delta$ would have been homotopic to a closed geodesic in a single surface and it would have showed $\delta$ wasn't a geodesic to begin with. Thus we conclude no short non-contractible simple closed geodesic can be created upon sewing. Since $\gamma = L_c$ and $\Sigma_1 \sharp \Sigma_2 \in \partial \mathcal{V}_{pl}(L_o,L_c)$, the planar string vertices $\mathcal{V}_{pl}(L_o,L_c)$ solve the geometric planar master equation~\eqref{eq:PlanarBV} as a result.
\end{proof}

We remind that the Feynman diagrams constructed using $\mathcal{V}_{pl}(L_o,L_c)$ cover the relevant moduli space once and only once, in the sense that they represent a fundamental class in the homology relative to the boundary when pushed-forward to the moduli space, similar to the general open-closed vertices. This still holds true thanks to Theorem 1 in~\cite{Cho:2019anu}. However, we don't know such Feynman regions define pieces of section over $\widehat{\mathcal{P}}^{0,n}_{b, \{m_i\}} \to \mathcal{M}^{0,n}_{b, \{m_i\}}$ in the case of hyperbolic vertices.

Before closing this subsection we remark on a minor modification to vertices of~\cite{Cho:2019anu}. In~\eqref{eq:HyperbolicVertices}, $\widetilde{\mathcal{V}}^{0,1}_{1, 0} (L_c)$ is defined to be a flat circle of circumference $L_c$. However, it is also possible to define this as a flat cylinder of circumference $L_c$ and length $s$. This change requires grafting a flat finite-sized cylinder (i.e. a \textit{stub}) of length $s$ to $b$-borders.\footnote{Refer~\cite{Chiaffrino:2021uyd,Schnabl:2023dbv,Erbin:2023hcs,upcoming_work} for recent treatises on stubs.} This modification turns out to be quite useful: it makes sure the divergences associated with the boundary state on $b-$borders is tamed.

The chain $\p \mathcal{V}$ still contains surfaces whose $b$-borders have length $L_c$ after this modification---adding stubs doesn't affect anything here. It can be shown that this part of $\p \mathcal{V}$ is still contained in the right-hand side of the master equation~\eqref{eq:BV}, given these surfaces can be obtained by $\{ \mathcal{V}^{0,1}_{1, 0}(L_c), \chi \}_c$ where $\chi$ is the same surface for which $b$-border, and its stub, is replaced with a $c$-border. Sewing new $\mathcal{V}^{0,1}_{1, 0}(L_c)$ to $\chi$ then generates the $b$-border and its stub. In the converse direction, sewing $\mathcal{V}^{0,1}_{1, 0}(L_c)$ to any other surface doesn't create a short $p$-geodesics like in~\cite{Cho:2019anu}. So these surfaces are in $\p \mathcal{V}$. From this moment forward we assume stubs of length $s$ always exist around the $b$-borders and they are included implicitly when a $b$-border has created upon sewing.

\subsection{The $L_o \to \infty, L_c \to 0$ limit}

Now we consider singular limits of planar hyperbolic string vertices obeying the constraint~\eqref{eq:Constraint}. Let us begin by making a comment on what happens when the border length associated with closed (open) strings approaches to $0$. We claim this limit is only meaningful if closed (open) strings are taken on-shell: it is no longer possible to associate local coordinates around the relevant punctures as grafted flat cylinders (strips) disappear in this limit. The only way to consistently insert bulk (boundary) vertex operators is through inserting ones for which the off-shell data is irrelevant, i.e. by taking them on-shell. By similar reasons, there can't be any Feynman diagrams with closed (open) string propagators in this limit. So taking $L_c \to 0$ ($L_o \to 0$) in hyperbolic vertices integrates out closed (open) string fields.\footnote{This statement is true up to boundary contributions of the moduli space corresponding to on-shell states produced inside the diagrams. We ignore them here, however they can be tracked systematically using the approach below.}

That being said, we now consider the limit
\begin{align} \label{eq:EllwoodLimit}
	L_o \to \infty, \quad L_c \to 0, \quad ds \to {\pi \over L_o} ds, \quad \sinh L_o \, \sinh L_c \leq 1\, .
\end{align}
We call this~\textit{the Ellwood limit} for reasons that would be apparent soon. Here $ds$ stands for the partially flat, partially hyperbolic metric of hyperbolic vertices and the third limit corresponds to infinite scaling one has to do after taking $L_o \to \infty$, see~\cite{Costello:2019fuh, Firat:2023glo}. Observe that we keep satisfying the geometric master equation~\eqref{eq:BV} as the limit is taken by the last condition.

We begin by investigating the limits of dimension-zero vertices, refer to appendix~\ref{app} for details. We have four zero-dimensional vertices in general. However $\mathcal{V}^{0,3}_{0,0}(L_c)$ and $\mathcal{V}^{0,1}_{1,0}(L_c)$ are not relevant as we are taking closed string on-shell, so we only have to consider $\mathcal{V}^{0,0}_{1,3}(L_o)$ and $\mathcal{V}^{0,1}_{1,1}(L_o, L_c)$. Recall that the closed string insertion in $\mathcal{V}^{0,1}_{1,1}(L_o, L_c)$ has to be on-shell. We claim $\mathcal{V}^{0,0}_{1,3}(L_o \to \infty)$ and $\mathcal{V}^{0,1}_{1,1}(L_o\to\infty, L_c\to 0)$ are Witten's vertex and the vertex associated with the Ellwood invariant, which we call~\textit{Ellwood vertex}, respectively. The local coordinates for these vertices are plotted in figure~\ref{fig:Ellwood}. The explicit expressions for the local coordinates can be found in~\cite{Ellwood:2008jh,rastelli2001tachyon}.
\begin{figure}[t]
	\includegraphics[width=0.50\textwidth,height=0.48\textwidth]{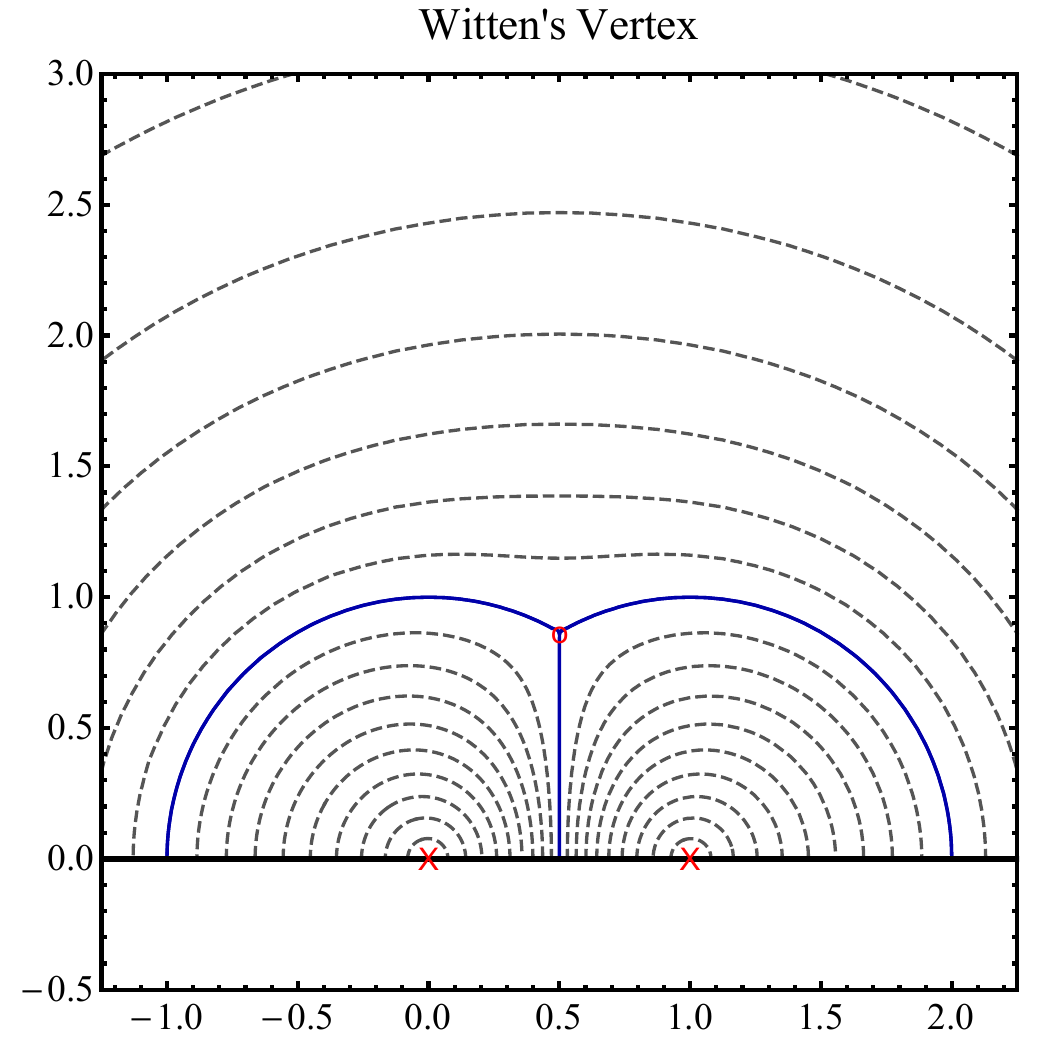}
	\includegraphics[width=0.50\textwidth,height=0.48\textwidth]{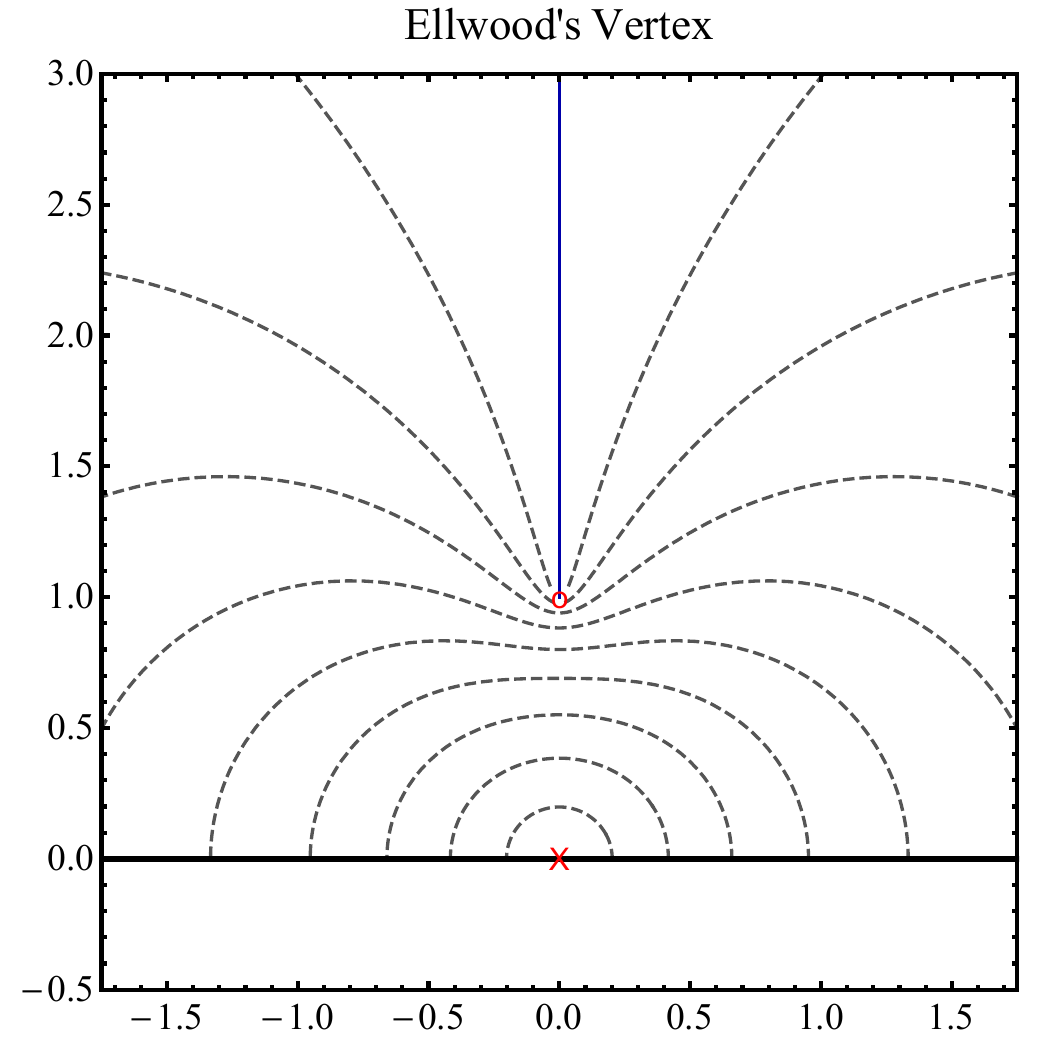}
	\caption{\label{fig:Ellwood}Witten's and Ellwood's vertices. We map contours $|w_i| = 0.1, \cdots, 1$ of the local coordinates (dashed lines) to the UHP $z$. These are the horizontal trajectories of the associated differentials~\eqref{eq:Ellwood}. The critical graph is the solid blue curve, the critical points (zeros and simple poles) are marked with red points, and the double poles are marked with a red cross. }
\end{figure}

In order to prove this claim, recall that there are Strebel differentials associated with Witten's and Ellwood's vertex. These are given by, eg. by using their local coordinates~\cite{Ellwood:2008jh,rastelli2001tachyon},
\begin{subequations} \label{eq:Ellwood}
	\begin{align}
	&\varphi_{W}= - \left[ {1 \over z^2} + {1 \over (1-z)^2} + {1 \over z} +{ -1 \over z-1}
	\right] dz^2 
	= - {z^2 - z + 1 \over z^2(1-z)^2} dz^2
	\,  , \\
	& \varphi_{E}   = -
	\left[
	{1 \over z^2} + {i \over 2(z-i)} - {i \over 2(z+i)} \right]dz^2 
	= - {1 \over z^2 (z^2+1)} dz^2 \, ,
	\end{align}
\end{subequations}
on the UHP $z$. These differentials contain entire information regarding the local coordinates around punctures. So the vertices $\mathcal{V}^{0,0}_{1,3}(L_o)$ and $\mathcal{V}^{0,1}_{1,1}(L_o,L_c)$ defined according to~\eqref{eq:HyperbolicVertices} can be shown to reduce Witten's and Ellwood's vertex by establishing the existence of Strebel differentials~\eqref{eq:Ellwood} underlying their local coordinates in the Ellwood limit~\eqref{eq:EllwoodLimit}.

This can be done by taking the WKB limit of the Fuchsian equation~\eqref{eq:Fuchsian} with the choices for $T(z)$ given as in~\eqref{eq:WKB}. Note that the WKB limit refers to the Ellwood limit~\eqref{eq:EllwoodLimit} in this context. As we highlighted earlier, the entire off-shell data regarding $\mathcal{V}^{0,0}_{1,3}(L_o)$ and $\mathcal{V}^{0,1}_{1,1}(L_o,L_c)$ are contained within the Fuchsian equation. Furthermore, as argued in~\cite{Firat:2023glo}, stress-energy tensors given in~\eqref{eq:WKB} result in the following Strebel differentials
\begin{subequations}
\begin{align}
	&-2 \pi^2 \lim_{\substack{L_o \to \infty \\ L_c \to 0}} {T(z; \mathcal{V}^{0,0}_{1,3}(L_{o})) \over L_o^2} dz^2 = \varphi_{W} \, ,\quad  -2 \pi^2 \lim_{\substack{L_o \to \infty \\ L_c \to 0}} {T(z; \mathcal{V}^{0,1}_{1,1}(L_o,L_c)) \over L_o^2} dz^2  = \varphi_{E} \, ,
\end{align}
\end{subequations}
associated with their local coordinates as $L_o \to \infty, L_c \to 0$. These are the same as those given in~\eqref{eq:Ellwood}, which establishes our claim. We point out having Strebel differentials makes sense in the Ellwood limit. Hyperbolic parts and $c$-borders shrink as $L_o \to \infty, L_c \to 0$ by the Gauss-Bonnet theorem and these produce a ``ribbon-like graph'' described by the critical graphs of Strebel differentials~\cite{Firat:2023glo}. This is apparent in figure~\ref{fig:Ellwood}: hyperbolic parts and $c$-borders collapses to the solid blue curves and all that remains are the grafted semi-infinite flat strips. We rescale the lengths by $\pi / L_o$, which is equivalent to measuring the lengths in the metrics defined by the Strebel differentials $\sqrt{|\varphi|}$. The lengths of the $p$-sides are equal to $\pi$ in this metric.

In fact, the Ellwood limit of hyperbolic open-closed SFT corresponds to the Witten's open SFT deformed by the Ellwood invariant~\cite{zwiebach1992interpolating,Ellwood:2008jh}. Witten's and Ellwood's vertex alone are known to cover~\textit{all} moduli spaces relevant to open-closed SFT with open string Feynman diagrams at the expense of taking closed strings on-shell~\cite{zwiebach1992interpolating}.  It is an instructive exercise to establish this fact purely within hyperbolic SFT. Like argued in~\cite{zwiebach1992interpolating}, it is sufficient to show dimension-one vertices are empty in the limit $L_o \to \infty, L_c \to 0$ to show all relevant moduli spaces are covered with open string Feynman diagrams constructed using Witten's and Ellwood's vertices. This is because open string propagators increase the dimension of vertices by one and the geometric master equation~\eqref{eq:BV} is assumed to keep being satisfied as the limit is taken, which is only possible as long as all higher vertices disappear if dimension-one vertices disappear in the process. This leads to the mentioned covering.

Thus we only have to show dimension-one vertices are absent. Decompositions of moduli spaces of these vertices are characterized in~\cite{Cho:2019anu} for generic values of $L_c,L_o$, see appendix~\ref{app} for summary. All we have to do at this point is then to take the Ellwood limit of the inequalities describing the vertex regions of dimension-one vertices while satisfying the geometric master equation. Case-by-case we have:
\begin{enumerate}
	\item \underline{Disk with four boundary punctures $\mathcal{V}^{0,0}_{1,4} (L_o \to \infty)$}. The inequality describing (a single piece of) vertex region takes the form of
	\begin{align}
		\pi = \pi \lim_{\substack{L_o \to \infty \\ L_c \to 0}} {L_o \over L_o} \leq \ell \leq
		\pi \lim_{\substack{L_o \to \infty \\ L_c \to 0}} { e(L_o) \over L_o} = \pi \, ,
	\end{align}
	using~\eqref{eq:A1}. Here, and in the subsequent analysis, $\ell$ denotes the length of the $p$-geodesic measured in the Strebel metric $ds = \sqrt{|\varphi|}$. We see $\mathcal{V}^{0,0}_{1,4} (L_o \to \infty) = \emptyset$ as a result, which makes sense: it is a well-known fact that Feynman diagrams constructed using Witten's vertex cover the entire moduli space $\mathcal{M}^{0,0}_{1,4}$~\cite{Zwiebach:1990az}.
	
	\item \underline{Disk with one bulk and two boundary punctures $\mathcal{V}^{0,1}_{1,2}(L_o \to \infty, L_c \to 0)$}. The inequality describing this vertex region takes the form of
	\begin{align}
	\pi = \pi \lim_{\substack{L_o \to \infty \\ L_c \to 0}} {L_o \over L_o} \leq \ell \leq
	\pi \lim_{\substack{L_o \to \infty \\ L_c \to 0}} { f(L_o, L_c) \over L_o} = \pi \, ,
	\end{align}
	using~\eqref{eq:A2}. Again, we see the vertex is absent $\mathcal{V}^{0,1}_{1,2} (L_o \to \infty ,L_c\to 0) = \emptyset$. For an alternative argument based on minimal-area metrics, see~\cite{zwiebach1992interpolating}.

	\item \underline{Disk with two bulk punctures $\mathcal{V}^{0,2}_{1,0}(L_o \to \infty,L_c \to 0)$}. Unlike the cases above here we find
	\begin{align} \label{eq:Amb}
	\pi \lim_{\substack{L_o \to \infty \\ L_c \to 0}} { g(L_c) \over L_o}  \geq 2 \pi \, ,
	\end{align}
	using~\eqref{eq:A3} for the upper-bound of the vertex region. Having an inequality means that this limit is ambigious and it depends on the relation between $L_c$ and $L_o$ as the Ellwood limit~\eqref{eq:EllwoodLimit} is taken. The inequality is saturated when $\sinh L_c \sinh L_o =1$. We expect that the Ellwood limit to be unique on the other hand (see figure~\ref{fig:Constraint}). So the only way to be consistent with this expectation is by interpreting this ambiguity as having no surfaces beyond $\ell > \pi$, i.e. $\ell = \pi$ to be the boundary of the moduli space. Under this interpretation this vertex is absent $\mathcal{V}^{0,2}_{1,0}(L_o \to \infty, L_c\to 0) = \emptyset$.
	
	This interpretation is justified when we recall the argument by Zwiebach~\cite{zwiebach1992interpolating}. As we have argued above, the vertex $\mathcal{V}^{0,1}_{1,1}(L_o \to \infty,L_c\to 0)$ is the Ellwood vertex. We can construct a Feynman diagram to describe a disk with two bulk punctures by identifying two Ellwood vertices by their boundary punctures. This is shown in figure~\ref{fig:sing}. As the open string propagator collapses ($s \to 0$) two bulk punctures approach each other. This corresponds to surface being singular and lying on the boundary of moduli space. From hyperbolic perspective $s = 0$ is equivalent to $\ell = \pi$ because any surface with $\ell < \pi$ contains an open string propagator by the $p$-systole condition~\eqref{eq:Cho}. So both of these frameworks suggest interpreting the ambiguity in the limit~\eqref{eq:Amb} as hitting the boundary of the moduli space.
	\begin{figure}[t]
		\centering
		\includegraphics[width=0.8\textwidth,height=0.32\textwidth]{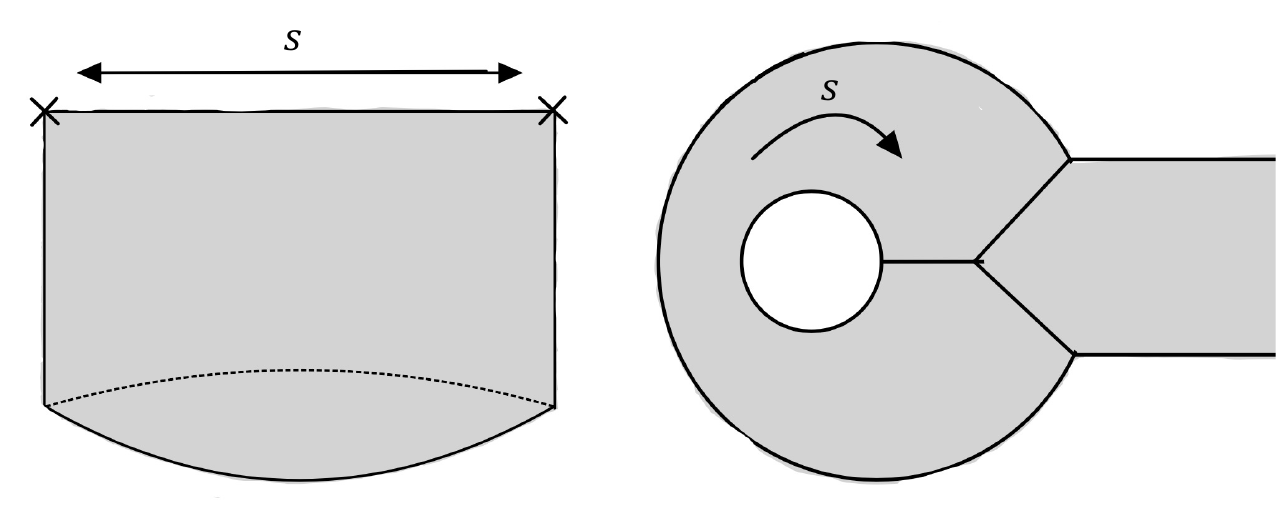}
		\caption{\label{fig:sing}Disk with two on-shell bulk punctures~(\textit{left}) and annulus with one boundary puncture~(\textit{right}) as Feynman diagrams in the Ellwood limit. Surfaces are singular when $s\to 0$.}
	\end{figure}
	
	Looking at surfaces with $p$-geodesic $L \geq L_o$ before taking the Ellwood limit provides another justification. These surfaces don't contain any open string propagator or any part whose length diverges with $L_o$. So they shrink to a point as $L_o \to \infty$, i.e. they become singular. The only way to be consistent is then taking corresponding region in the moduli space to be absent.
	
	\item \underline{Annulus with one boundary puncture $\mathcal{V}^{0,0}_{2,\{1,0\}}(L_o \to \infty, L_c \to 0)$}. Similar to the case above, we also have an ambiguity for the limit of the upper-bound of the vertex region
	\begin{align}
		\pi \lim_{\substack{L_o \to \infty \\ L_c \to 0}} { h(L_o, L_c) \over L_o}  \geq {3\pi \over 2} \, , 
	\end{align}
	by~\eqref{eq:A4}. Again, we interpret this as hitting the boundary of the moduli space at $\ell = \pi$ and $\mathcal{V}^{0,0}_{2,\{1,0\}}(L_o \to \infty,L_c \to 0) = \emptyset$ as a result. We can argue this is consistent with~\cite{zwiebach1992interpolating} like in $\mathcal{V}^{0,2}_{1,0}(L_o \to \infty, L_c \to 0)$, see figure~\ref{fig:sing}. Alternatively, we can consider surfaces with $p$-geodesics $L \geq L_o$ before taking the limit and observe these surfaces become singular in the limit.

\end{enumerate}

Before closing off this subsection we would like to point out a subtlety.\footnote{The author thanks Carlo Maccaferri for emphasizing this point.} As we noted, the consistency of the Ellwood limit requires taking closed strings on-shell. However this is~\text{not} sufficient for the construction here: the state has to be (weight $(0,0)$) primary as well. In order to see why, recall that the level-matched closed string cohomology supports an on-shell non-primary state, the ghost-dilaton $D$, controlling the value of string coupling constant. This is a zero-momentum state, together with its transformation law $w \to f(w)$, given by
\begin{align}
D = {1 \over 2} ( c \p^2 c - \overline{c} \overline{\p}^2 \overline{c} ) \, , \quad \quad
f \circ D = D - {f''(0) \over 2 f'(0)^2} \, c \partial c + {\bar{f}''(0) \over 2 \bar{f}'(0)^2} \, \bar{c} \bar{\partial} \bar{c} \, .
\end{align}
This state leads to an amplitude diverging as $\sim L_o^2$ as $L_o \to \infty$ when inserted into $\mathcal{V}^{0,1}_{1,1}(L_o,L_c)$, inspect equation (1.3) in~\cite{Firat:2021ukc} for intuition. To avoid such divergence we are forced to take closed string insertions primary. Consistent coupling of the ghost-dilaton to the open SFT is an important problem and most likely requires regularizing this divergence while keeping track of the boundary term contributions of the moduli space carefully. We plan to discuss this issue elsewhere.

\subsection{The $L_o \to 0, L_c \to \infty$ limit} \label{sec:MRV}

Now we consider the ``opposite'' limit where open strings are taken on-shell instead:
\begin{align} \label{eq:MRV}
	L_o \to 0 \, , \quad
	L_c \to \infty \, , \quad
	ds \to {2\pi \over L_c} ds \, , \quad
	\sinh L_o \, \sinh L_c \leq 1 \, .
\end{align}
This limit is rather exotic and it describes closed SFT deformed by on-shell open strings ending on $N$ identical D-branes when $N$ is large. Importantly, notice that this situation can only be realized consistently in the large $N$ open-closed SFT by the last condition on the border lengths---not in generic open-closed SFT, see figure~\ref{fig:Constraint}. We denote this limit~\textit{the Maccaferri-Ruffino-Vo\v{s}mera (MRV) limit} after~\cite{Maccaferri:2023gof}.

Our objective in this subsection is to investigate the local coordinates resulting from the MRV limit and the associated covering of the moduli spaces of planar Riemann surfaces. Since we take $ L_c \to \infty$ we produce polyhedral vertices based on Strebel differentials. Simultaneously we take $L_o \to 0$, so the local coordinates around boundary punctures disappear entirely. This requires taking open strings on-shell. In this case every open string boundary becomes a simple closed geodesic of the surface with some markings on it and these can be effectively treated as $b$-borders. Measured in the Strebel metric $\sqrt{|\varphi|}$ the length of a $c$-border is equal to $2\pi$.

In light of this remark, let us investigate the dimension-zero and dimension-one vertices. Since open strings are on-shell we don't need to consider $\mathcal{V}^{0,0}_{1,3}(L_o)$. The only relevant vertices for this dimension are $\mathcal{V}^{0,3}_{0,0}(L_c)$, $\mathcal{V}^{0,1}_{1,0}(L_c)$, $\mathcal{V}^{0,1}_{1,1}(L_o,L_c)$. In the MRV limit, the local coordinates for $\mathcal{V}^{0,3}_{0,0}(L_c\to \infty)$ are given by the closed string version of Witten's vertex~\cite{Costello:2019fuh,Firat:2021ukc} while the local coordinates for $\mathcal{V}^{0,1}_{1,1}(L_o\to 0,L_c\to\infty)$ is simply given by~\textit{the Cayley map}
\begin{align} \label{eq:Cayley}
w(z) = {z-i \over z+i} \implies \varphi_C = - {dw^2 \over w^2} = {4 \, dz^2 \over (1+z^2)^2} \, ,
\end{align}
that maps punctured unit disk $0 < |w| \leq 1$ to the UHP $z$ (up to an overall phase ambiguity), see figure~\ref{fig:Deformation}. Having this coordinate is a simple consequence of having no grafted flat strip around the boundary puncture. 
\begin{figure}[t]
	\includegraphics[width=0.50\textwidth,height=0.48\textwidth]{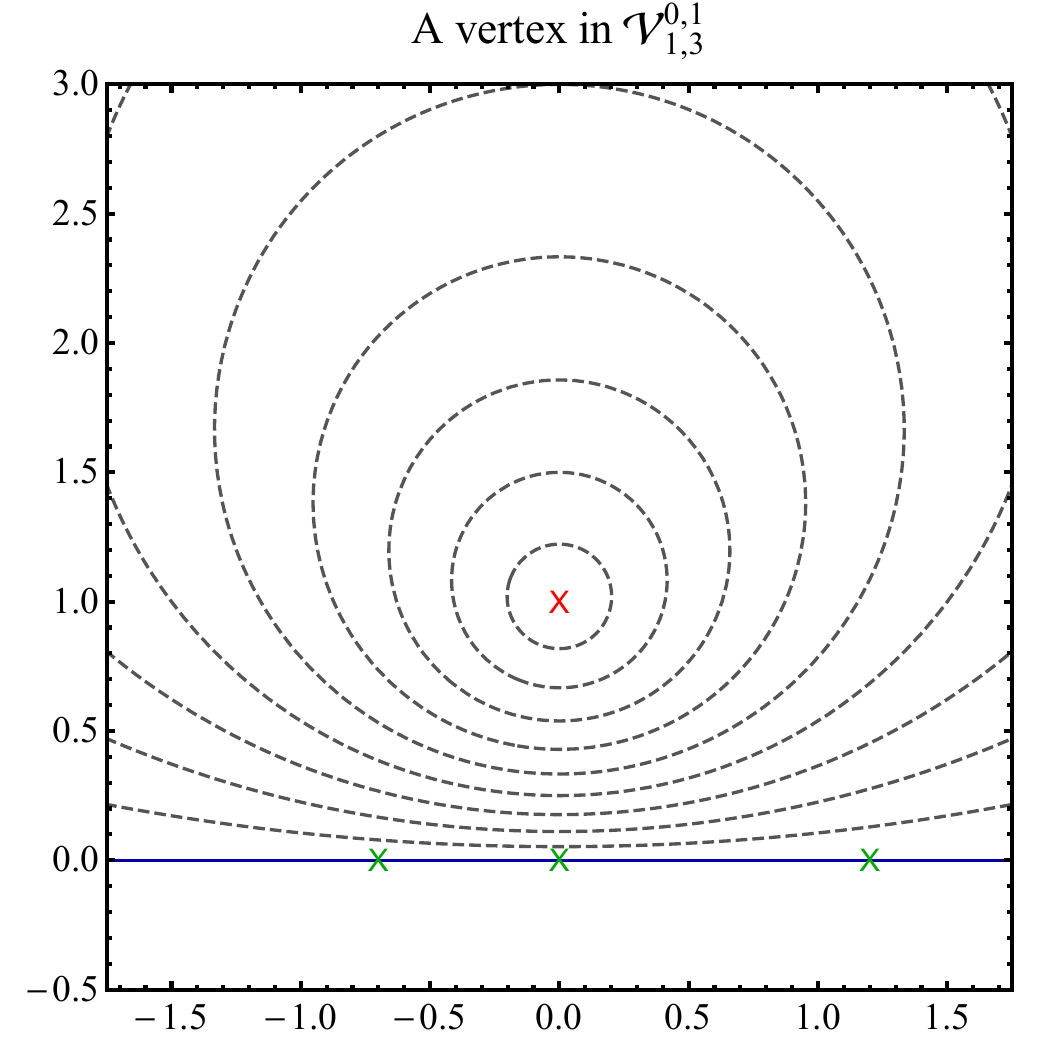}
	\includegraphics[width=0.50\textwidth,height=0.48\textwidth]{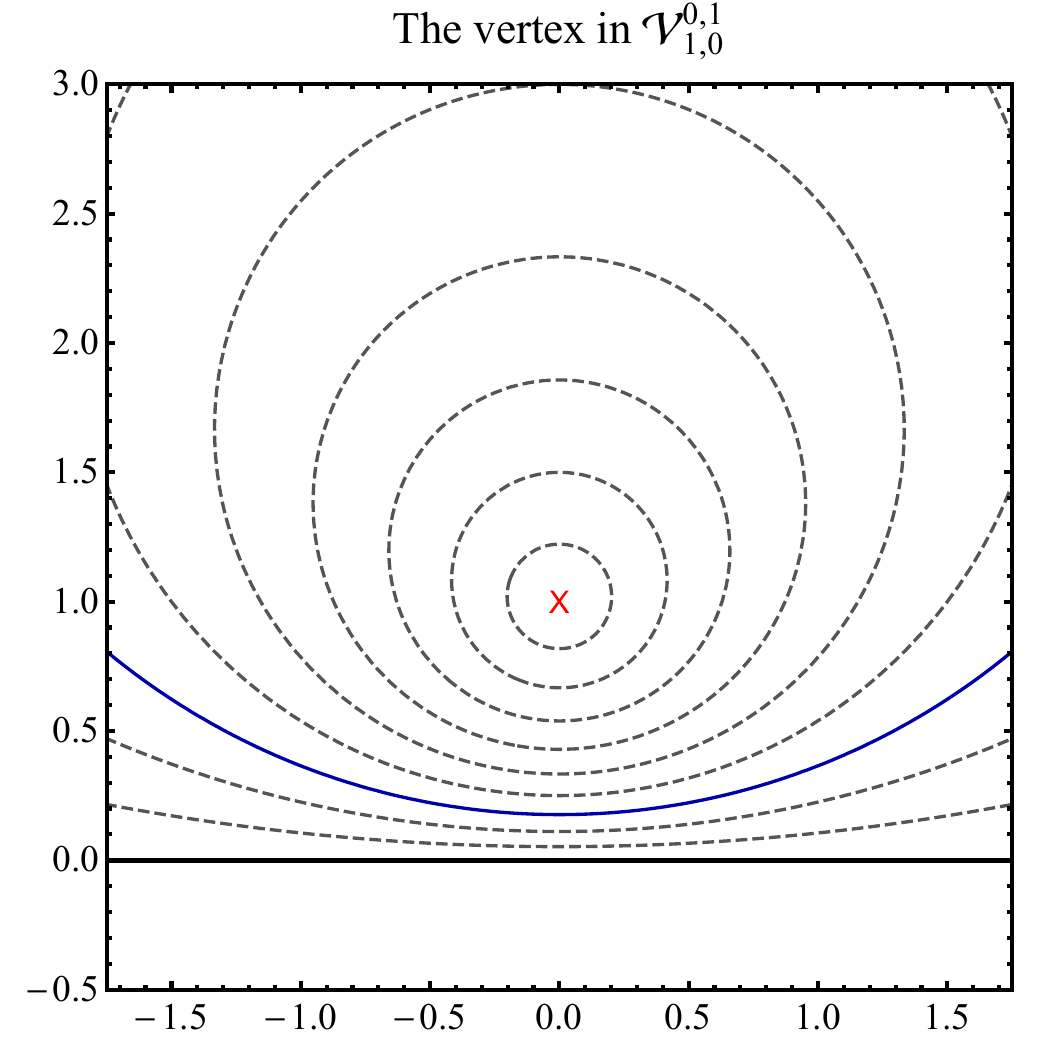}
	\caption{\label{fig:Deformation}Examples of vertices $\mathcal{V}^{0,1}_{1,m}(L_o \to 0,L_c \to \infty)$. The conventions from earlier figure apply. We distinguish the on-shell boundary punctures by marking them with green cross. Notice for the vertex $\mathcal{V}_{1,0}^{0,1}(L_c\to \infty)$ we have taken the flat stub (of length $s = -\log 0.7 \approx 0.36$) into account: $|w|=1$ maps to the blue solid curve.}
\end{figure}
By the same logic, the Cayley map also gives the local coordinates for $\mathcal{V}^{0,1}_{1,m}(L_o\to 0,L_c\to\infty)$---except for $m=0$ where the stub around the $b-$border has to taken account by scaling $w(z) \to e^{-s} w(z)$ additionally. Notice the moduli for these vertices are just the relative positions of boundary punctures and the associated Strebel differential is given by~\eqref{eq:Cayley}.

We point out the Strebel differential $\varphi_C$~\eqref{eq:Cayley} matches with the one coming from the WKB limit of Fuchsian equation with the stress-energy tensor~\eqref{eq:DiskT}
\begin{align}
	-4 \pi^2 \lim_{\substack{L_o \to 0 \\ L_c \to \infty}} {T(z; \mathcal{V}^{0,1}_{1,m}(L_{o},L_{c})) \over L_c^2} dz^2 = \varphi_{C} \, ,
\end{align}
where the WKB limit refers to the MRV limit~\eqref{eq:MRV} in this context. Here a crucial observation behind this derivation is that the accessory parameters $c_j^R$ don't scale with $L_c$, as we don't want them to overwhelm the double poles around boundary punctures while the limit is taken~\cite{Firat:2023glo}. This shows the local coordinates around the bulk puncture are indeed given by the Cayley map~\eqref{eq:Cayley}. 

\begin{figure}[t]
	\includegraphics[width=0.50\textwidth,height=0.463\textwidth]{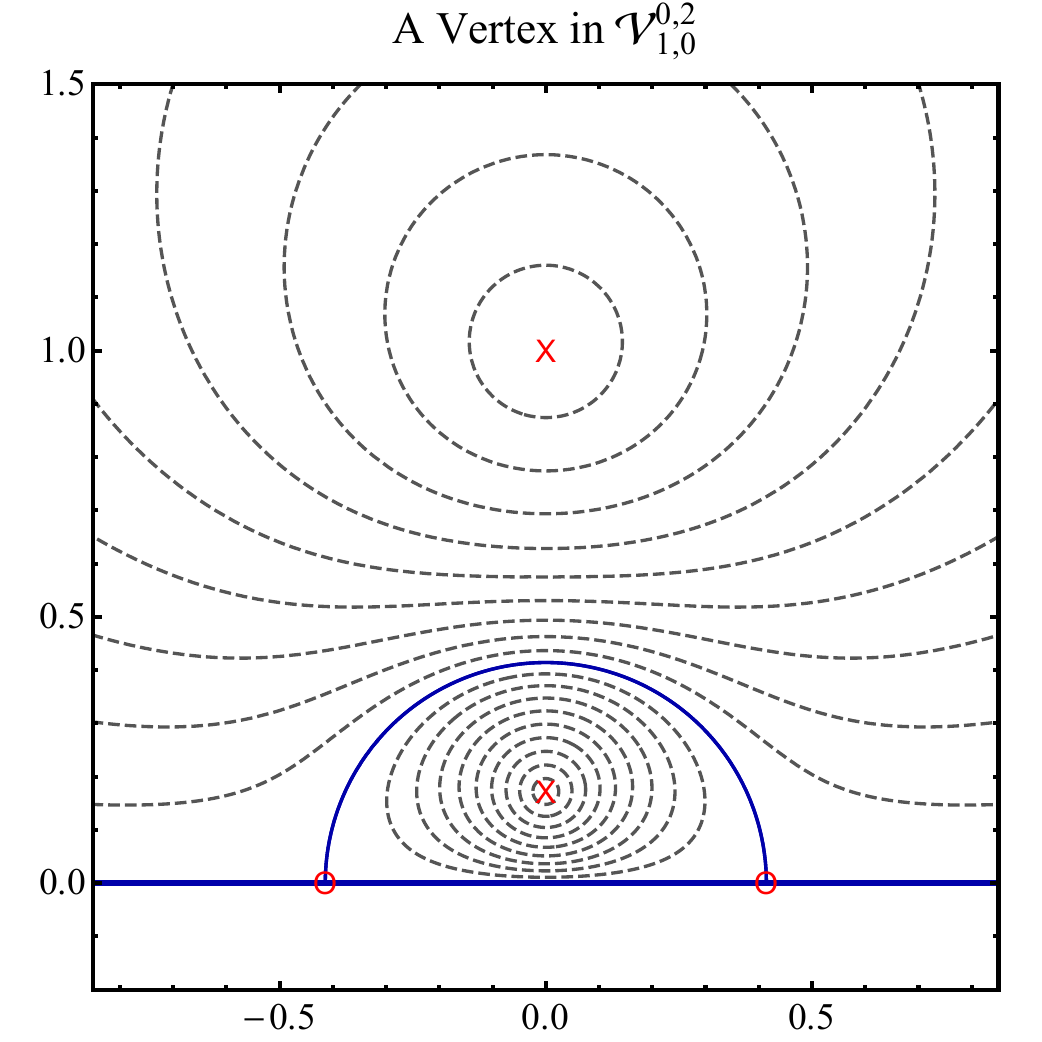}
	\includegraphics[width=0.50\textwidth,height=0.46\textwidth]{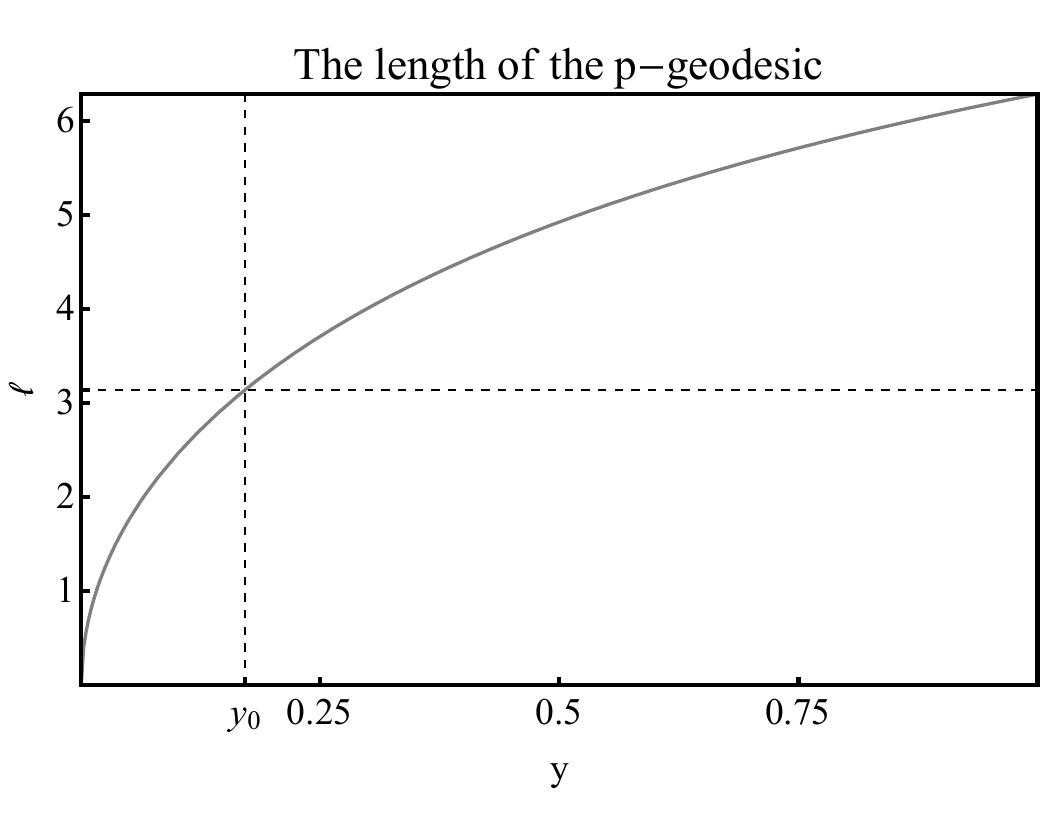}
	\caption{\label{fig:Deformationc}An example of a vertex $\mathcal{V}^{0,2}_{1,0}(L_o \to 0,L_c \to \infty)$~(\textit{left}) and the dependence of the length of $p$-geodesic (blue half-circle) $\ell$ measured in the Strebel metric to the position of the puncture~(\textit{right}). The conventions from earlier figure apply. The punctures are placed at $z=i, i (3-2\sqrt{2} )$. The vertex region is given by $0 \leq \ell \leq \pi$, or equivalently $0 \leq y \leq y_0 = 3-2\sqrt{2}$. }
\end{figure}
Next, we investigate dimension-one vertices relevant for the MRV limit case-by-case:
\newline

\begin{enumerate}
	\item \underline{Disk with one bulk and two boundary punctures $\mathcal{V}^{0,1}_{1,2}(L_o\to0,L_c\to\infty)$}. We find an ambiguity in the limit
	\begin{align}
		2 \pi \lim_{\substack{L_o \to 0 \\ L_c \to \infty}} { f(L_o, L_c) \over L_c} \geq  10\pi \, .
	\end{align}
	 Assuming the MRV limit is unique, this is an indication of hitting the boundary of the moduli space at the upper-bound of the vertex region parameterized by $\ell$. Observe that the lower-bound of the vertex region is given by $\ell = 0$. So $\mathcal{V}^{0,1}_{1,2}(L_o\to 0,L_c\to \infty)$ has to encapsulate the entire moduli space, $\mathcal{V}^{0,1}_{1,2}(L_o\to 0,L_c\to \infty) = \mathcal{M}^{0,1}_{1,2}$. We can reach the same conclusion by observing there isn't any Feynman diagram with a closed string propagator for this topology. So every surface should be part of the vertex region given open strings are integrated out.
	
	\item \underline{Disk with two bulk punctures $\mathcal{V}^{0,2}_{1,0}(L_o \to 0,L_c \to \infty)$}.  The inequality describing the vertex region for this topology takes the form
	\begin{align}
	0 = 2 \pi \lim_{\substack{L_o \to 0 \\ L_c \to \infty}} {L_o \over L_c} \leq \ell \leq
	2 \pi \lim_{\substack{L_o \to 0 \\ L_c \to \infty}} { g(L_c) \over L_c} = \pi \, ,
	\end{align}
	so we have $\mathcal{V}^{0,2}_{1,0}(L_o \to 0,L_c \to \infty) \neq \emptyset$. This should be a polyhedral vertex on the disk with two bulk punctures as there can't be any internal flat strip in the geometry given open strings are integrated out. This means there is underlying Strebel differentials for the surfaces in this vertex region. This can be shown to be given by
	\begin{align}
		\varphi \left(\mathcal{V}^{0,2}_{1,0}(L_o \to 0,L_c \to \infty)\right) = {4 (y-1)^2 (z^2 -y)^2 \over (z^2+1)^2 (z^2 +y^2)^2} dz^2 \, ,
	\end{align}
	after placing the punctures at $z=i, iy$ for $0 \leq y \leq y_o = 3 - 2 \sqrt{2}$. The expressions for the local coordinates around $z=i, iy$ (up to an overall phase ambiguity) are given by
	\begin{align}
	w_0 (z) = - \left({i \sqrt{y} - 1 \over i \sqrt{y} +1} \right)^2 { (z + i y)(z-i) \over (z-iy) (z+i)} \, , \quad
	w_1 (z) =  - \left({i \sqrt{y} - 1 \over i \sqrt{y} +1} \right)^2 { (z - i y)(z+i) \over (z+iy) (z-i)} \, ,
	\end{align}
	respectively, see figure~\ref{fig:Deformationc}. Refer to appendix~\ref{appB} for derivations.
	
	\item \underline{Annulus with one boundary puncture $\mathcal{V}^{0,0}_{2,\{1,0\}}(L_o \to 0,L_c \to \infty)$}. The inequality describing the vertex region for this topology takes the form
	\begin{align}
		0 = 2 \pi \lim_{\substack{L_o \to 0 \\ L_c \to \infty}} {L_o \over L_c} \leq \ell \leq
		2 \pi \lim_{\substack{L_o \to 0 \\ L_c \to \infty}} { h(L_o, L_c) \over L_c} = 0 \, ,
	\end{align}
	so we have $\mathcal{V}^{0,0}_{2,\{1,0\}}(L_o \to 0,L_c \to \infty) = \emptyset$. The vanishing of this vertex region can be understood as a consequence of infinite scaling, which makes surfaces with $p$-geodesics of length $L_o\leq L \leq h(L_o,L_c)$ before taking limit singular after taking the limit. In fact the same logic can be applied to argue $\mathcal{V}^{0,0}_{b,\{m_i\}}(L_o \to 0,L_c \to \infty) = \emptyset$.
	
\end{enumerate}

Despite $\mathcal{V}^{0,0}_{b,\{m_i\}}(L_o \to 0,L_c \to \infty) = \emptyset$, having $\mathcal{V}^{0,2}_{1,0}(L_o \to 0,L_c \to \infty) \neq \emptyset$ indicates the covering of moduli spaces doesn't truncate as in the Ellwood limit: there are still infinitely-many vertex regions associated with  planar Riemann surfaces with at least one bulk puncture. Maybe this shouldn't be too surprising given that the closed SFT part is already non-polynomial. In any case, we find that the hyperbolic vertices in the MRV limit becomes polyhedral. In other words, their off-shell data are characterized by Strebel differentials.

Let us close this subsection by remarking on two important points. Similar to the ghost-dilaton insertions of the Ellwood limit, taking open strings on-shell is not sufficient for the MRV limit: we also have to take them to be weight-$0$ primary. So inserting Nakanishi-Lautrup field $\p c$, for example, may possibly lead to divergences that require a careful treatment than what is provided here. Secondly, like in~\cite{Maccaferri:2023gof}, we obtain closed SFT with a tadpole after taking the MRV limit. These tadpoles have to eliminated to obtain the back-reacted closed string background. This would take a non-trivial effort, however, as it requires solving closed SFT equation of motions, which is notoriously difficult problem.

\section{Conclusion} \label{sec:conclusion}

In this note we constructed a set of string vertices for open-closed SFT in the large $N$ limit based on hyperbolic geometry by relaxing a condition in~\cite{Cho:2019anu}. In particular we show that there is a remarkably simple geometric picture based on Strebel differentials after taking either open or closed strings on-shell. We argued corresponding limits essentially implement integrating-out open or closed strings in the context of hyperbolic open-closed SFT. 

The open-closed (super-)string field theory in the large $N$ limit, if non-perturbatively well-defined,\footnote{There is an indication this is a quite reasonable ``if'', see~\cite{upcoming_work}.} is expected to play a role in the proof of AdS/CFT correspondence. Therefore it is imperative to understand this theory in various regimes as much as possible. Our parametrization gives a family of vertices that can interpolate between ``bulk'' ($L_c \to \infty, L_o \to 0$) and ``boundary'' ($L_c \to 0, L_o \to \infty$) perspectives seamlessly. We point out that the boundary point of view was suggested to characterize the dual closed string theory by Okawa~\cite{Okawa:2020llq}. However, the opposite direction, i.e. using the bulk point of view to characterize the dual open string theory, hasn't been explored in great detail in the literature.

It is worth mentioning that the interpolation from closed to open strings above corresponds to a particular field redefinition within large $N$ open-closed SFT~\cite{Hata:1993gf}. Assuming SFT is a complete theory of string theory, this field redefinition is expected to map degrees of freedom in the gravity side to those of the gauge theory side. This is intriguing, as it will provide a novel field-theoretical approach to the holographic principle within string theory. However the success of applying this supposed SFT-based framework to AdS/CFT boils down to the amount of theoretical control we can exert on hyperbolic string vertices. Continuing the line of research~\cite{Firat:2021ukc,Firat:2023glo,Firat:2023suh}, the author hopes to make further progress towards this objective in future.

Finally, as we demonstrated explicitly, Strebel differentials make a natural appearance in the ``bulk'' point of view of open-closed SFT in the large $N$ limit. This type of quadratic differentials has been discovered to be relevant using entirely different set of reasoning from gauge/gravity correspondence~\cite{gopakumar2011simplest,gopakumar2005free}. It would be interesting to shed light on their connection.

\section*{Acknowledgments}
The author thanks Harold Erbin, Carlo Maccaferri, and Barton Zwiebach for their insightful comments on the early draft and discussions; and Ted Erler, Daniel Harlow, and Manki Kim for stimulating discussions on related topics. This material is based upon work supported by the U.S. Department of Energy, Office of Science, Office of High Energy Physics of U.S. Department of Energy under grant Contract Number  DE-SC0012567.

\appendix
\section{Low dimensional vertices} \label{app}

\begin{figure}[t]
	\centering\includegraphics[width=0.6\textheight,height=0.29\textheight]{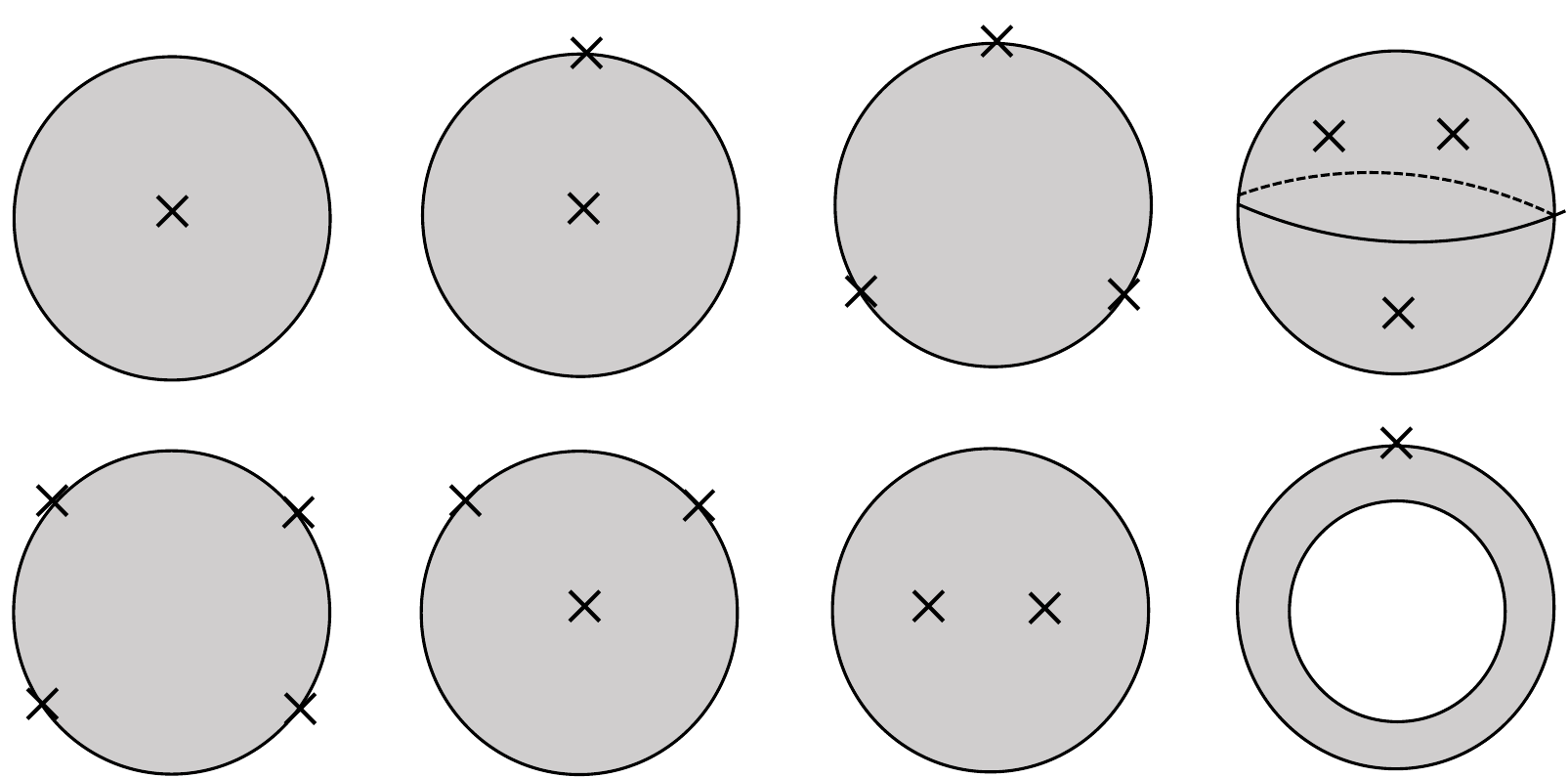}
	\caption{\label{fig:low}Topologies relevant to dimension-zero~(\textit{top}) and dimension-one~(\textit{bottom}) vertices.}
\end{figure}
In this appendix we characterize dimension-one and dimension-zero hyperbolic vertices in terms of the hyperbolic lengths $L$ of one of their $p$-geodesic. This appendix summarizes the results of~\cite{Cho:2019anu} and included for the convenience of the reader. Let us begin by considering dimension-zero vertices shown~in figure~\ref{fig:low}. We have four dimension-zero vertices:
\begin{enumerate}
	\item \underline{Disk with one bulk puncture $\mathcal{V}^{0,1}_{1,0}(L_c)$},
	\item \underline{Disk with one bulk and one boundary punctures $\mathcal{V}^{0,1}_{1,1}(L_o,L_c)$},
	\item \underline{Disk with three boundary punctures $\mathcal{V}^{0,0}_{1,3}(L_o)$},
	\item \underline{Sphere with three bulk punctures $\mathcal{V}^{0,3}_{0,0}(L_c)$}.
\end{enumerate}

Likewise, there are four different dimension-one vertices, which are also shown in figure~\ref{fig:low}. The decompositions associated with these vertices are given as follows:
\begin{figure}
	\centering
	\includegraphics[width=0.7\textheight,height=0.3\textheight]{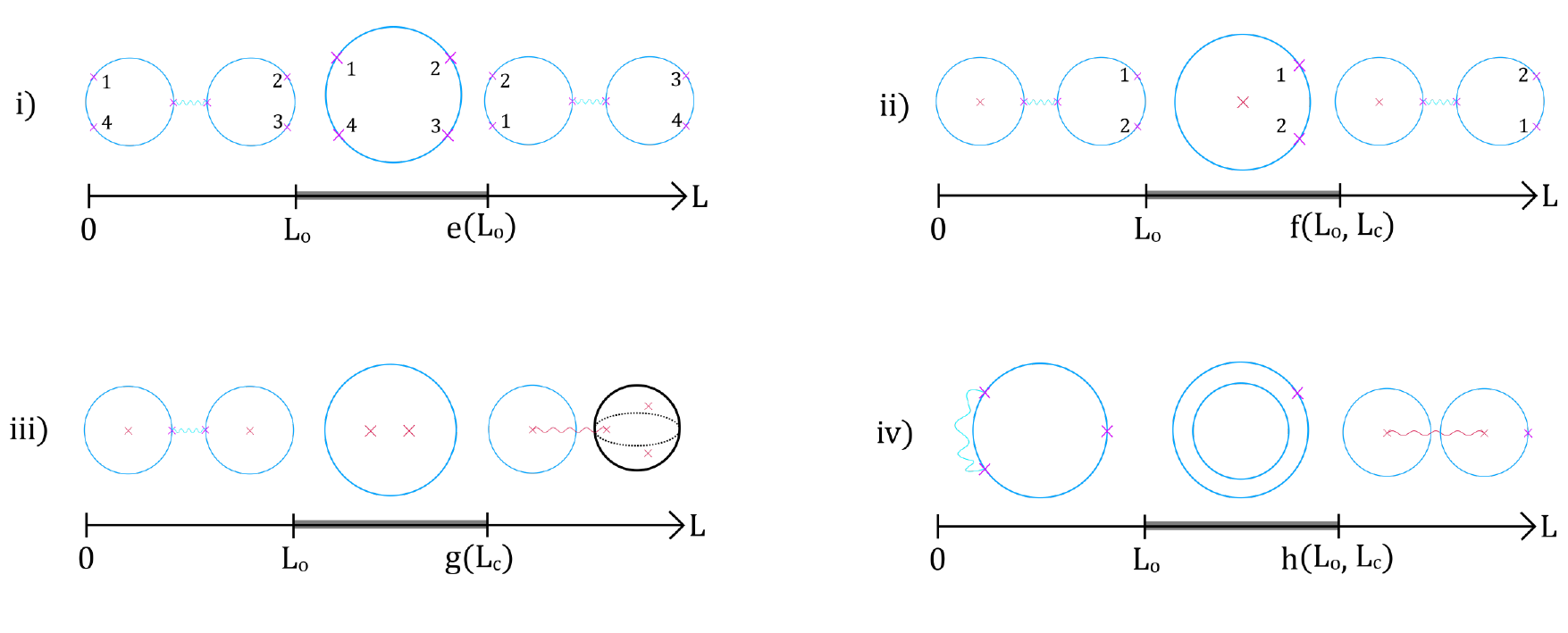}
	\caption{\label{fig:one}The decomposition of one-dimensional moduli spaces of hyperbolic open-closed string vertices. The figure is taken from~\cite{Cho:2019anu}. The author thanks Minjae Cho for the permission to use this figure.}
\end{figure}
\begin{enumerate}	
	\item \underline{Disk with four boundary punctures $\mathcal{V}^{0,0}_{1,4}(L_o)$}. The vertex region in this case is characterized by the inequality $L_o \leq L \leq e(L_o)$ where
	\begin{align} \label{eq:A1}
		\cosh (e (L_o)) = {2 \cosh(2L_o) + \cosh L_o + 1 \over \cosh L_o - 1} \, .
	\end{align}
	The regions $0 < L < L_o$ and $e(L_o) < L$ corresponds to the Feynman regions. Unlike the rest of the cases below, there are actually five more copies for these regions in the moduli space $\mathcal{M}^{0,0}_{1,4}(L_o)$ due to the permutation of boundary punctures. We only consider one copy for simplicity. 
	
	\item \underline{Disk with one bulk and two boundary punctures $\mathcal{V}^{0,1}_{1,2}(L_o,L_c)$}. The vertex region in this case is characterized by the inequality $L_o \leq L \leq f(L_o, L_c)$ where
	\begin{align} \label{eq:A2}
		\cosh(f(L_o,L_c)) = - 1 &+  2 \coth^2 {L_o \over 2}\bigg[
		\sqrt{\cosh L_c + 1 \over \cosh L_o -1} \cosh L_o \\
		&+ \sqrt{ (\cosh L_o + \cosh L_c)(\coth^2 L_o \coth^2{L_o \over 2}-1) \over \cosh L_o + 1} \sinh L_o
		\bigg]^2 \, . \nonumber
	\end{align}
	The regions $0 < L < L_o$ and  $f(L_o, L_c) < L $ correspond to the Feynman regions.
	
	\item \underline{Disk with two bulk punctures $\mathcal{V}^{0,2}_{1,0}(L_o, L_c)$}.  The vertex region in this case is characterized by the inequality $L_o \leq L \leq g(L_c)$ where
	\begin{align} \label{eq:A3}
	\cosh (g(L_c)) = {\cosh^2 {L_c \over 4} + \cosh L_c \over \cosh^2 {L_c \over 4} - 1} \, .
	\end{align}
	The regions $0 < L < L_o$ and $e(L_c) < L $ correspond to the Feynman regions.
	
	\item \underline{Annulus with one boundary puncture $\mathcal{V}^{0,0}_{2,\{1,0\}}(L_o, L_c)$}. The vertex region in this case is characterized by the inequality $L_o \leq L \leq h(L_o, L_c)$ where
	\begin{align} \label{eq:A4}
		\cosh (h(L_o,L_c)) = \sqrt{ \cosh L_c + \cosh L_o \over \cosh L_c - 1} \, .
	\end{align}
	The regions $0 < L < L_o$ and  $h(L_o, L_c) < L $ correspond to the Feynman regions.
\end{enumerate}
The schematics for these decompositions are shown in figure~\ref{fig:one}.

\section{The local coordinates for $\mathcal{V}^{0,2}_{1,0}(L_o \to 0, L_c \to \infty)$} \label{appB}

In this appendix we investigate the local coordinates for $\mathcal{V}^{0,2}_{1,0}(L_o\to 0,L_c \to \infty)$ on the UHP with uniformizing coordinates $z$ in detail and characterize its vertex region. To that end, place the bulk punctures at $z=i, iy$ using $PSL(2,\mathbb{R})$ freedom, where $0 < y < 1$. Notice $1 < y < \infty$ can be related to this situation by inversion. Next, perform the doubling trick around the real axis and replace the UHP with two bulk punctures with a complex plane with four bulk punctures. The punctures are placed at $z=i, iy, -iy, -i$ after doubling. Notice they lie on a line.

As mentioned in the main text, the local coordinates $\mathcal{V}^{0,2}_{1,0}(L_o\to 0,L_c \to \infty)$ should be described by Strebel differentials. After the doubling, the question becomes finding the Strebel differential on four-punctured sphere whose punctures at $z=i, iy, -iy, -i$. Luckily, this problem has already solved in~\cite{Erbin:2022rgx}---to translate the we need to conformally map $z=i, -iy, -i$ to $z'=0,1,\infty$. This can be done via the M\"obius map
\begin{align} \label{eq:map}
	z' = f(z) =  {(y-1)(z-i) \over (y+1) (z +i)}
	\, .
\end{align}
Observe that the puncture $iy$ is getting mapped to
\begin{align} \label{eq:xiy}
	f(iy) = {(y-1)^2 \over (y+1)^2 } \, ,
\end{align}
in the $z'$-plane. Given these facts, we find the Strebel differential on the $z$-plane to be
\begin{align} \label{eq:PS}
	\varphi \left(\mathcal{V}^{0,2}_{1,0}(L_o \to 0, L_c \to \infty)\right) = {4 (y-1)^2 (z^2 -y)^2 \over (z^2+1)^2 (z^2 +y^2)^2} dz^2 \, ,
\end{align}
after pulling it back the quadratic differential (2.14) of~\cite{Erbin:2022rgx} with the map $f$~\eqref{eq:map}. It can be checked that~\eqref{eq:PS} has double poles at $z=i, iy, -iy, -i$ with residue equal to $-1$ and can be written as a sum of its singular term around these punctures. Notice the zeros of~\eqref{eq:PS} are second order and they are located at $z = \pm \sqrt{y}$.

The local coordinates around the punctures associated with this differential can be found by
\begin{align}
	- {dw_i ^2 \over w_i^2} = \varphi \left(\mathcal{V}^{0,2}_{1,0}(L_o \to 0,L_c \to \infty)\right) \, .
\end{align}
Here $w_i$ for $i = 0,1$ are the local coordinates around $z= i, i y$ respectively. Remarkably, we can integrate~\eqref{eq:PS} exactly and find (up to an overall phase ambiguity) 
\begin{align}
	w_0 (z) = - \left({i \sqrt{y} - 1 \over i \sqrt{y} +1} \right)^2 { (z + i y)(z-i) \over (z-iy) (z+i)} \, , \quad
	w_1 (z) =  - \left({i \sqrt{y} - 1 \over i \sqrt{y} +1} \right)^2 { (z - i y)(z+i) \over (z+iy) (z-i)} \, ,
\end{align}
where we used the known positions of zeros of~\eqref{eq:PS} to fix the integration constants. The local coordinates for $y=1/2$ is shown in figure~\ref{fig:Deformationc}. We point out the critical graph of~\eqref{eq:PS} on the UHP is a union of a half-circle and the real line, where the latter corresponds to the boundary of the disk.

The lengths of the non-contractible simple geodesics as functions of the cross-ratio are also found in~\cite{Erbin:2022rgx}, see equation (2.17). We can use this expression to relate the length $\ell$ of the $p$-geodesic (which is a half-circle on the UHP) to $y$ using~\eqref{eq:xiy} and find
\begin{align}
	\ell = 4 \arctan \left[ {2 \sqrt{y} \over 1- y } \right] \, ,
\end{align}
see figure~\ref{fig:Deformationc}. From here we immediately see the vertex region for $\mathcal{V}^{0,2}_{1,0}(L_o\to 0,L_c \to \infty)$ is given by the inequality $0 \leq y \leq y_0 \equiv 3 - 2 \sqrt{2}$. This concludes our explicit characterization of the local coordinates and the vertex region.


\providecommand{\href}[2]{#2}\begingroup\raggedright\endgroup

\end{document}